\newif\ifreport\reporttrue
\documentclass[journal]{IEEEtran}
\usepackage{setspace}
\usepackage{cite}
\usepackage{graphicx}
\usepackage{caption}
\usepackage{subcaption}
\usepackage[cmex10]{amsmath}
\usepackage{mathtools}
\usepackage{amsthm}
\usepackage{amsfonts}
\usepackage{amssymb}
\usepackage{algorithm}
\usepackage{algorithmic}
\usepackage{bm,xstring,hyperref}
\usepackage{fixltx2e}%fixes latex float algorithm in 2 columns
\usepackage{tikz}
\usetikzlibrary{decorations.pathreplacing}
\newcommand{\EE}{\mathbb{E}}

\usepackage{color}

\def\blue{\color{black}}

\newcommand{\ignore}[1]{}
\usepackage[singlelinecheck=false]{caption}

\newtheorem{lemma}{Lemma}

\newtheorem{theorem}{Theorem}
\newtheorem{corollary}{Corollary}
\theoremstyle{definition}

\begin{document}
%\title{Optimizing the Accuracy of Real-Time Monitoring}
%\title{Optimal Control of Real-Time Monitoring Systems}
%\title{Optimal Dynamic Sampling for Real-Time Monitoring of Brownian Motion}
%\title{Real-Time Updates of Brownian Motion }
%\title{Optimal Sampling of Brownian Motion }
\title{Optimal Sampling of Brownian Motion for Real-time Monitoring} % and Its Relation with the Age of Information}
% Real-time Monitoring}
%\title{Beyond the Age-of-information: Real-time Sampling and Estimation of Brownian Motion} % and Its Relation with the Age of Information}
\title{Beyond the Age-of-information: Causal Sampling and  Estimation of Brownian Motion} % and Its Relation with the Age of Information}

\title{From Age-of-information to MMSE: Real-time Sampling of Brownian Motion} % and Its Relation with the Age of Information}

\title{Age-of-Information and Tracking of Wiener Process over Channel with Random Delay}

\title{Remote Estimation of the Wiener Process over a Channel with Random Delay}

%\title{From Age-of-information to MMSE: Causal Sampling and Estimation of Brownian Motion} 

%\title{Dynamic Sampling of Brownian Motion for }

%\title{Optimal Sampling of Brownian Motion in a Random Channel}

\IEEEoverridecommandlockouts
%\title{Update or Wait: How to Keep Your Data Fresh}
\author{Yin Sun, Yury Polyanskiy, and Elif Uysal-Biyikoglu\\
Dept. of ECE,  the Ohio State University, Columbus, OH\\
Dept. of EECS, Massachusetts Institute of Technology, Cambridge, MA\\
Dept. of EEE, Middle East Technical University, Ankara, Turkey\\
~\\
May 22, 2017

\thanks{The research was supported in part by the Center for Science of Information (CSoI),
an NSF Science and Technology Center, under grant agreement CCF-09-39370, by ONR grant N00014-17-1-2417, and by TUBITAK.}
}
\maketitle
% !TEX root = ./sampling_BM.tex
\begin{abstract}

In this paper, we consider a problem of sampling a Wiener process, with samples forwarded to a remote estimator via a channel that consists of a queue with random delay. The estimator reconstructs a real-time estimate of the signal from causally received samples. Motivated by recent research on age-of-information, we study the optimal sampling strategy that minimizes the mean square estimation error subject to a sampling frequency constraint. We prove that the optimal sampling strategy is a threshold policy, and find the optimal threshold. This threshold is determined by the sampling frequency constraint and how much the Wiener process  varies during the channel delay. %Hence, the optimal sampling policy is determined by the features of the signal, sampler, and channel in a simple form.
An interesting consequence is that even in the absence of the sampling frequency constraint, the optimal strategy is \emph{not} zero-wait sampling in which a new sample is taken once the previous sample is delivered; {\blue rather, it is optimal to wait for a non-zero amount of time after the previous sample is delivered, and then take the next sample.} Further, if the sampling times are independent of the observed Wiener process, the optimal sampling problem reduces to an age-of-information optimization problem that has been recently solved. Our comparisons show that the estimation error of the optimal sampling policy is much smaller than those of age-optimal sampling, zero-wait sampling, and  classic uniform sampling. 

\end{abstract}

% !TEX root = ./sampling_BM.tex
\section{Introduction}

{\blue 
Consider a system with two terminals (see Fig.~\ref{fig_model}): An observer measuring a Wiener process $W_t$ and an estimator, whose goal is to provide the best-guess $\hat W_t$ for the current value of $W_t$. These two terminals are connected by a channel 
that transmits time-stamped samples of the form $(S_i, W_{S_i})$, where the sampling times $S_i$ satisfy $0\leq S_1 \leq S_2\leq\ldots$ The channel is modeled as a work-conserving FIFO queue with 
%any given non-preemptive service order discipline (e.g., FCFS or LCFS) and 
random \emph{i.i.d.}  delay $Y_i$, where $Y_i\geq 0$ is the channel delay (i.e., transmission time) of sample $i$.\footnote{
%This channel model is similar to that considered in \cite{Anantharam1996}. 
By ``work-conserving'', we meant that the channel is kept busy whenever there exist some generated samples that are not delivered to the estimator.
} The observer can choose the sampling times $S_i$ causally subject to an average sampling frequency constraint
\begin{align} %\label{eq_rate_constraint}
\liminf_{n\to\infty} {1\over n} \mathbb{E}[S_n] \ge {1\over f_{\max}}, \nonumber
\end{align}
where $f_{\max}$ is the maximum allowed sampling frequency. }

Unless it arrives at an empty system, sample $i$ needs to wait  in the queue until its transmission starts. Let $G_i$ be the transmission starting time of sample $i$ such that $S_i \leq G_i$. 
%For any non-preemptive queueing discipline, 
The delivery time of sample $i$ is $D_i = G_i + Y_i$. The initial value $W_0= 0$ is known by the estimator for free, represented by $S_0 =D_0 = 0$. At time $t$, the estimator forms $\hat W_t$ using causally received samples with $D_i \le t$. By minimum mean square error (MMSE) estimation,  %\cite{Schervish1995}
\begin{align}\label{eq_esti}
\hat{W}_{t}  = & \mathbb{E}[{W}_t | W_{S_j}, D_j \leq t ] \nonumber\\
                = & W_{S_i}, ~\text{if}~t\in[D_i,D_{i+1}),~i=0,1,2,\ldots, 
\end{align}
as illustrated in Fig. \ref{fig_signal}.
We measure the quality of remote estimation via the MMSE:
\begin{align}%\label{eq_MMSE}
\limsup_{T\to \infty} {1\over T} \EE \left[\int_0^T (W_t - \hat W_t)^2 dt\right].\nonumber
\end{align}

%The optimal sampling problem for minimizing the MMSE subject to the  sampling frequency constraint is formulated as 

% by solving the following optimal sampling problem.

\begin{figure}%[!t]
\centering
%{\subfigure[][The zero-wait policy.]{\resizebox{0.45\textwidth}{!}{\includegraphics{./matlab_SY/example1}}}}
  %  \hspace{0.1\textwidth}
\includegraphics[width=0.4\textwidth]{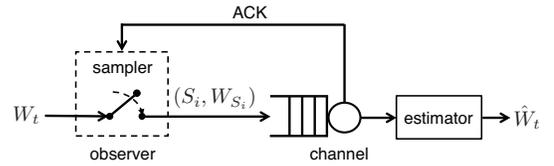}   
\caption{System model.}\vspace{-0.0cm}
\label{fig_model}
\end{figure}    

%\subsection{Contributions} 
In this paper, we study  the optimal sampling strategy that achieves the fundamental tradeoff between $f_{\max}$ and MMSE.
The contributions of this paper are summarized as follows:
\ifreport
\begin{itemize}
\else
\begin{itemize}[leftmargin=*]
\fi
\item The optimal sampling problem for minimizing the MMSE subject to the sampling frequency constraint  is solved exactly. We prove that the optimal sampling strategy is a threshold policy, and find the optimal threshold. This threshold is determined by $f_{\max}$ and the amount of signal variation during the channel delay (i.e., random transmission time of a sample). Our threshold policy has an important difference from the previous threshold policies studied in, e.g., \cite{Astrom2002,Hajek2008,imer2010, Nuno2011,nayyar2013, wu2013,Basar2014,Chakravorty2015,GaoCDC2015, GaoACC2016}: {\blue In our model, each sample waiting in the queue for its transmission opportunity unnecessarily becomes stale. We have proven that it is suboptimal to take a new sample when the channel is busy. Consequently, the threshold should be \emph{disabled} whenever there is a packet in transmission.}

%{\blue In our model, the samples may need to wait  in the queue for their transmission opportunity, and unnecessarily become stale while waiting. We prove that it is suboptimal to take a new sample when the channel is busy transmitting previous samples.  Consequently, the threshold is disabled whenever there is a packet in service.}

%

%As we have mentioned above, the threshold policy in this paper is  different from those of \cite{nayyar2013, wu2013,GaoCDC2015, GaoACC2016,Astrom2002,Basar2014,Chakravorty2015} in the following sense: To avoid the negative effect of queueing, the threshold-based control is disabled until all previous samples are delivered and  is reactivated after that. 

%One of our interesting findings is that the queue should be kept empty at all time such that a new sample is taken only after all previous samples are delivered.

\item An unexpected consequence of our study is that even in the absence of the sampling frequency constraint (i.e., $f_{\max}=\infty$), the optimal strategy is \emph{not} zero-wait sampling in which a new sample is generated once the previous sample is delivered; {\blue rather, it is optimal to wait a positive amount of time after the previous sample is delivered, and then take the next sample.}

\item If the sampling times are independent of the observed Wiener  process, the optimal sampling problem reduces to an age-of-information optimization problem solved in \cite{SunInfocom2016,report_AgeOfInfo2016}. The asymptotics of the MMSE-optimal and  age-optimal sampling policies at  low/high channel delay or low/high sampling frequencies  are studied. 

\begin{figure}[t!]
    \centering
%    \begin{subfigure}[t]{0.5\textwidth}
%        \centering
%        \includegraphics[width=0.8\textwidth]{example1}
%        \caption{System model.}\label{fig_model}
%    \end{subfigure}\\
    \begin{subfigure}[t]{0.5\textwidth}
        \centering
        \includegraphics[width=0.6\textwidth]{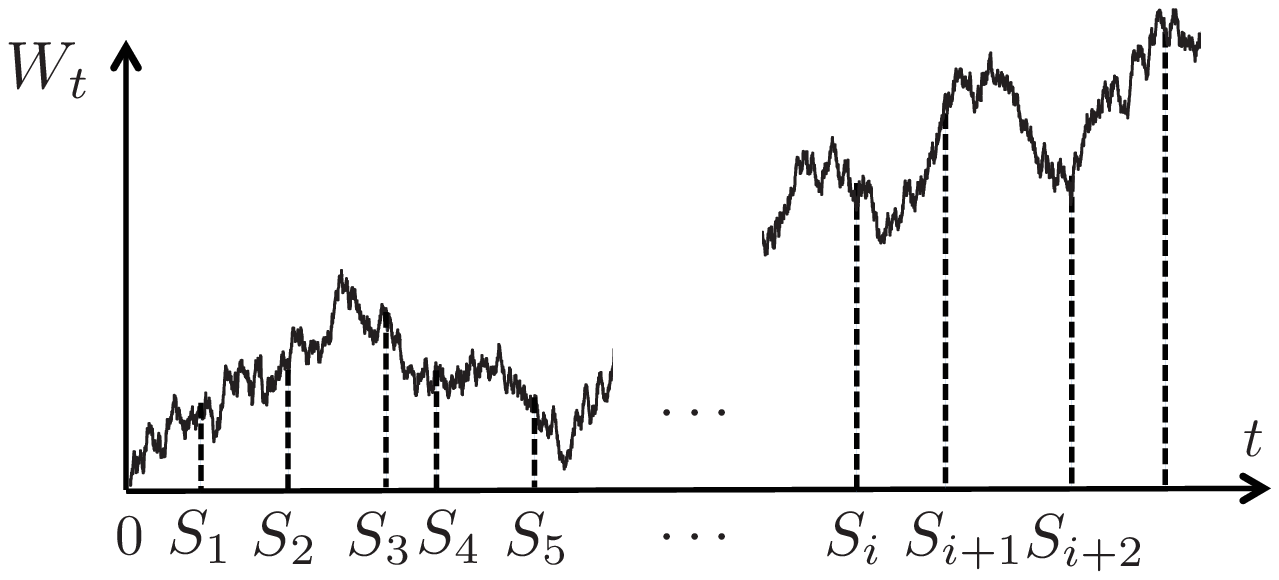}
        \caption{Wiener process $W_{t}$ and its samples.}\label{fig_signal1}
    \end{subfigure}\\
     %\vspace{0.5em}
    \begin{subfigure}[t]{0.5\textwidth}
        \centering
        \includegraphics[width=0.6\textwidth]{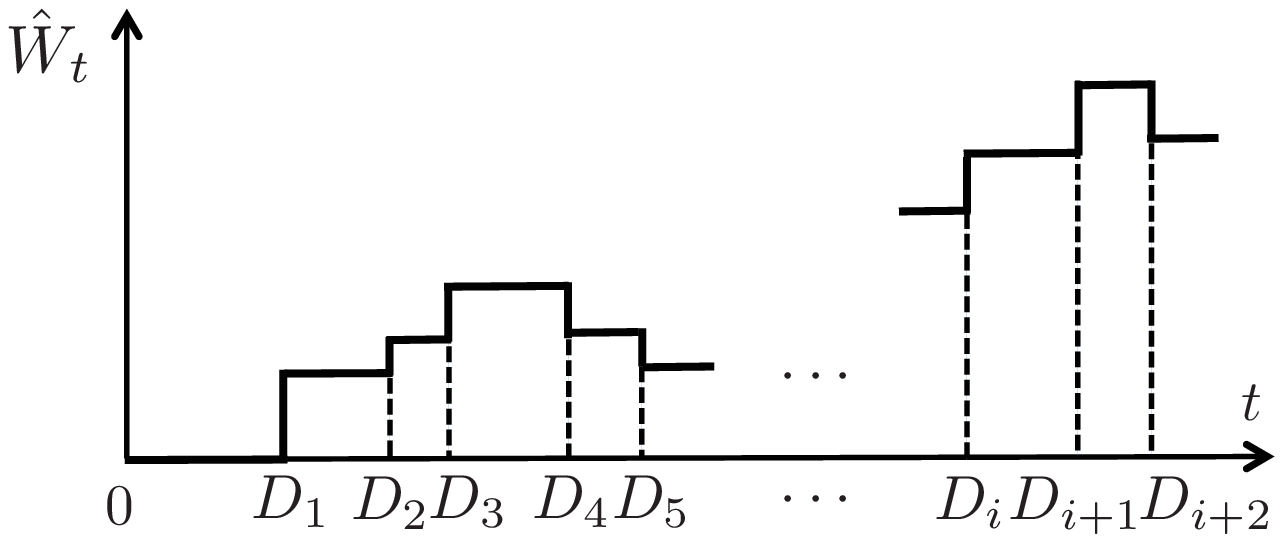}
        \caption{Estimate process $\hat{W}_{t}$ using causally received samples.}\label{fig_signal2}
    \end{subfigure}
                
    \caption{Sampling and remote estimation of the Wiener  process.}
    \ifreport
    \else
\fi\label{fig_signal}
\end{figure}

\item Our theoretical and numerical comparisons show that the MMSE of the optimal sampling policy is much smaller than those of age-optimal sampling, zero-wait sampling, and  classic uniform sampling. 
\end{itemize}

% !TEX root = ./sampling_BM.tex
\section{Related Work}

On the one hand, the results in this paper are closely related to the recent age-of-information studies, e.g., \cite{2015ISITYates,KaulYatesGruteser-Infocom2012,Kam-status-ISIT2013,Bacinoglu2015,SunInfocom2016,report_AgeOfInfo2016,Bedewy2016,Bedewy2017,Kadota2016,Bacinoglu2017}, where the focus was on queueing and channel delay, without a signal model. 
The discovery that the zero-wait policy is not always optimal for minimizing the age-of-information can be found in \cite{2015ISITYates,SunInfocom2016,report_AgeOfInfo2016}. The sub-optimality of a work-conserving scheduling policy was also observed in
\cite{Kadota2016}, which considered scheduling updates to different users with unreliable channels. 
%It is found that an age-optimal scheduling policy is throughput-optimal, while the converse is not true. 
One important observation in our study is that the behavior of the optimal update policy changes dramatically after adding a signal model.

On the other hand, the paper can be considered as a contribution to the rich literature on remote estimation, e.g., \cite{nayyar2013, wu2013,GaoCDC2015, GaoACC2016,Astrom2002,Basar2014,Chakravorty2015,mehra76, Hajek2008, Nuno2011, imer2010}, by adding a queueing model. Optimal transmission scheduling of sensor measurements for estimating a stochastic process was recently studied in \cite{GaoCDC2015, GaoACC2016}, where the samples are transmitted over a channel with additive noise. In the absence of channel delay and queueing (i.e., $Y_i=0$), the problems of sampling Wiener process and Gaussian random walk were addressed in \cite{Astrom2002,Basar2014,Chakravorty2015}, where the optimality of threshold policies was established. To the best of our knowledge, \cite{Basar2014} is the closest study with this paper. Because there is no queueing and channel delay in \cite{Basar2014}, the problem analyzed therein is a special case of ours. 

\section{Main Result}
%The system starts to operate at time $t=0$. 

%At any time $t$, the samples that are available to the estimator are $\{W_{S_j}: D_j \leq t \}$ 

Let $\pi=(S_0,S_1,\ldots)$ represent a sampling policy, and $\Pi$ be the set of \emph{causal} sampling policies %We consider a class of 
%\emph{causal} sampling policies $\Pi$ 
which satisfy the following conditions: (i) The \emph{information} that is available for determining the sampling time $S_i$ includes the history of the Wiener  process $(W_{t}: t\in[0, S_i])$, the history of channel states $(I_t: t\in[0, S_i])$, and the sampling times of previous samples $(S_0,\ldots,S_{i-1})$, 
where $I_t \in\{0,1\}$ is the idle/busy state of the channel at time $t$. (ii) The inter-sampling times $\{T_i = S_{i+1}-S_i, i=0,1,\ldots\}$ form a \emph{regenerative process} \cite[Section 6.1]{Haas2002}: There exist integers $0\leq {k_1}<k_2< \ldots$ such that the post-${k_j}$ process $\{T_{k_j+i}, i=0,1,\ldots\}$ has the same distribution as the post-${k_1}$ process $\{T_{k_1+i}, i=0,1,\ldots\}$ and is independent of the pre-$k_j$ process $\{T_{i}, i=0,1,\ldots, k_j-1\}$; in addition, $\mathbb{E}[S_{k_{1}}^2]<\infty$ and $0<\mathbb{E}[(S_{k_{j+1}}-S_{k_j})^2]<\infty$  for  $j=1,2,\ldots$\footnote{Really, we assume that $T_i$ is a regenerative process because we analyze the time-average MMSE in \eqref{eq_DPExpected}, but operationally a nicer definition is $\limsup_{n\rightarrow \infty}{\mathbb{E}[\int_0^{D_n} (W_{t}-\hat{W}_{t})^2dt]}/{\mathbb{E}[D_n]}$. These two definitions  are equivalent when $T_i$ is a regenerative process.} 
\ifreport
We assume that the Wiener process $W_t$ and the channel delay $Y_i$ are
determined by two external processes, which are mutually independent and  
do not change according to the sampling policy $\pi\in\Pi$. We also assume that the $Y_i$'s are \emph{i.i.d.} with $\mathbb{E}[Y_i^2]<\infty$. 
\else
We assume that the Wiener process $W_t$ and the channel delay $Y_i$ are mutually independent and  
do not change according to the sampling policy. We also assume $\mathbb{E}[Y_i^2]<\infty$.

\fi
\begin{figure}[t!]
    \!\!\!\!\begin{subfigure}[t]{0.23\textwidth}
        \centering
        \includegraphics[width=0.9\textwidth]{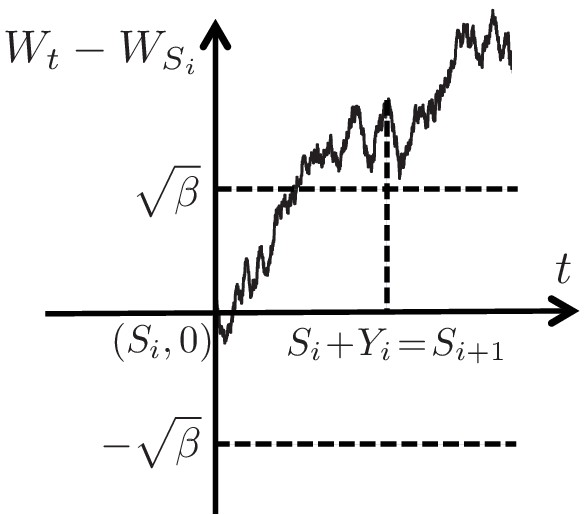}
            \vspace{-0.5em}
        \caption[(i)]{If $|W_{S_i+Y_i}\! -\! W_{S_i}|\! \geq\! \sqrt{\beta}$,  sample $i+1$ is taken at  time $S_{i+1} = S_i+Y_i$.}\label{}
    \end{subfigure}~~~
    \begin{subfigure}[t]{0.23\textwidth}
        \centering
        \includegraphics[width=0.9\textwidth]{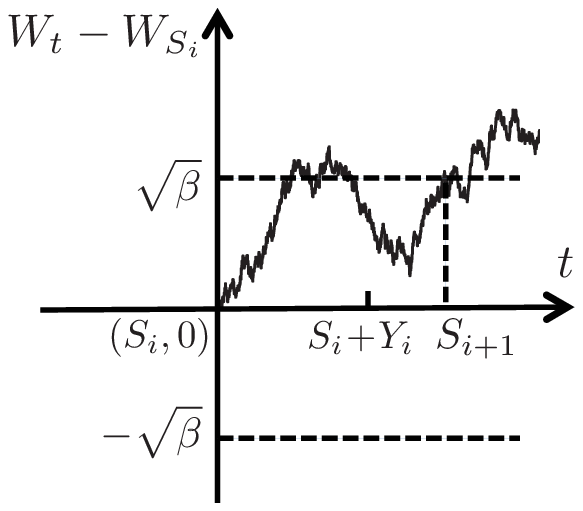}
            \vspace{-0.5em}
        \caption{If $|W_{S_i+Y_i} \!-\! W_{S_i}|\! <\! \sqrt{\beta}$, sample $i+1$ is taken at time $t$ that satisfies $t\geq S_i+Y_i$ and $|W_{t}\! -\! W_{S_i}|=\sqrt{\beta}$.}\label{}
    \end{subfigure}
  
    \caption{Illustration of the threshold policy \eqref{eq_opt_solution}.}\label{fig_policy1}
    \vspace{-1.5em}
\end{figure}

A sampling policy $\pi\in\Pi$ is said to be \emph{signal-independent} (\emph{signal-dependent}), if $\pi$ is (not) independent of the Wiener process $\{W_t,t\geq0\}$. Example policies in $\Pi$ include:
\begin{itemize}
\vspace{-0.5ex}
\ifreport 
\item[1.] \emph{Uniform sampling} \cite{Nyquist1928,Shannon1949}: 
\else
\item[1.] \emph{Uniform sampling}: 
\fi
The inter-sampling times are constant, such that for some $\beta\geq 0$,%sampling times are determined by   
\begin{align} \label{eq_uniform}
S_{i+1} = S_i+ \beta.
\end{align}
 
\item[2.] \emph{Zero-wait sampling} \cite{2015ISITYates,SunInfocom2016,report_AgeOfInfo2016,KaulYatesGruteser-Infocom2012}: A new sample is generated once the previous sample is delivered, i.e.,
\begin{align} \label{eq_Zero_wait}
S_{i+1} = S_i+ Y_i.
\end{align}

\item[3.] \emph{Threshold policy in signal variation:} 
The sampling times are given by
% when the variant of the Wiener process is no smaller than a threshold such that 
%The sampling time is determined by a threshold: 
%The sampling policy in \eqref{eq_thm2_1} and \eqref{eq_thm2}, which is the optimal solution to  \eqref{eq_age}.
\begin{align}\label{eq_opt_solution}
S_{i+1}= \inf \left\{ t\geq S_i + Y_i: |W_t - W_{S_i}| \!\geq\! \sqrt{\beta}\right\},
\end{align}
which is illustrated in Fig. \ref{fig_policy1}.
If $|W_{S_i+Y_i} - W_{S_i}| \geq \sqrt{\beta}$, 
%i.e., the amount of signal variation during the transmission of sample $i$ is larger than a threshold $\sqrt{\beta}$, 
sample $i+1$ is generated at the time $S_{i+1} = S_i+Y_i$ when sample $i$ is delivered; otherwise, if $|W_{S_i+Y_i} - W_{S_i}| < \sqrt{\beta}$, sample $i+1$ is generated at the earliest time $t$ such that $t\geq S_i+Y_i$ and %for an additional time period after $S_i+Y_i$ and take the next sample $i+1$ when 
$|W_t - W_{S_i}|$ reaches the threshold $\sqrt{\beta}$. {\blue It is worthwhile to emphasize that even if there exists time $t\in[S_i, S_i+Y_i)$ such that $|W_{t} - W_{S_i}| \geq \sqrt{\beta}$, no sample is taken at such time $t$, as depicted in Fig. \ref{fig_policy1}. 
In other words, the threshold-based control is disabled during $[S_i, S_i+Y_i)$ and is reactivated at time $S_i+Y_i$. 
%In other words, no new sample is generated until all previous samples are delivered, which avoids wasting time to wait in the queue. 
This is a key difference from previous studies on threshold policies \cite{nayyar2013, wu2013,GaoCDC2015, GaoACC2016,Astrom2002,Hajek2008,imer2010, Nuno2011,Basar2014,Chakravorty2015}.}
%the amount of signal variation has reached the threshold $\sqrt{\beta}$, i.e., $S_{i+1}\!>\! S_i+Y_i$ and $|W_{S_{i+1}} - W_{S_i}| = \sqrt{\beta}$.

\item[4.] \emph{Threshold policy in time variation} \cite{2015ISITYates,SunInfocom2016,report_AgeOfInfo2016}: The sampling times are given by
% when the variant of the Wiener process is no smaller than a threshold such that 
%The sampling time is determined by a threshold: 
%The sampling policy in \eqref{eq_thm2_1} and \eqref{eq_thm2}, which is the optimal solution to  \eqref{eq_age}.
\begin{align}\label{eq_thm2_1}
S_{i+1}=\inf \left\{ t\geq S_i + Y_i: t - {S_i} \geq{\beta}\right\}.
\end{align}

%\item[3.] \emph{Constant-variation sampling:} The sampling policy in \eqref{eq_thm2_1} with the average sampling frequency $f_{\max}$. 

\end{itemize}

The optimal sampling problem for minimizing the MMSE  subject to a  sampling frequency constraint is formulated as 
\begin{align}\label{eq_DPExpected}
%\pi_{\text{opt}}=\arg
&\min_{\pi\in\Pi}~ \limsup_{T\rightarrow \infty}\frac{1}{T}\mathbb{E}\left[\int_0^{T} (W_t - \hat W_t)^2dt\right] \\
&~\text{s.t.}~~ \liminf_{n\rightarrow \infty} \frac{1}{n} 
\mathbb{E}[S_n]\geq \frac{1}{f_{\max}}.\label{eq_constraint}
\end{align}

Problem  \eqref{eq_DPExpected} is a constrained continuous-time Markov decision problem with a continuous state space. 
Somewhat to our surprise, we were able to exactly solve \eqref{eq_DPExpected}:
%Furthermore, an unexpected consequence is that even in the absence of the sampling frequency constraint \eqref{eq_constraint}, the optimal strategy is not zero-wait sampling.
\begin{theorem}\label{thm_1}
There exists $\beta\geq0$ such that the sampling policy \eqref{eq_opt_solution} is optimal to \eqref{eq_DPExpected}, and the optimal $\beta$ is determined by solving\footnote{If $\beta\rightarrow0$, the last terms in \eqref{eq_thm1} and \eqref{eq_thm2} are determined by L'H\^{o}pital's rule.}
\begin{align} \label{eq_thm1}
\mathbb{E}[\max(\beta,W_Y^2)] \!=\! \max\left(\frac{1}{f_{\max}}, \frac{\mathbb{E}[\max(\beta^2,W_Y^4)]}{2\beta}\right),\!
\end{align}
where $Y$ is a random variable with the same distribution as $Y_i$.
%if $\beta>0$ in equation \eqref{eq_thm1}; otherwise, $\beta = 0$.
The optimal  value of \eqref{eq_DPExpected} is then given by \emph{
\begin{align}\label{thm_1_obj}
\mathsf{mmse}_{\text{opt}}\triangleq\frac{\mathbb{E}[\max(\beta^2,W_Y^4)]}{6\mathbb{E}[\max(\beta,W_Y^2)]} + \mathbb{E}[Y].
\end{align}}
%\!\!In addition, the expectations in \eqref{eq_thm1} and \eqref{thm_1_obj} are given by
%\begin{align}
%\mathbb{E}[\max(\beta,Y)] &= \mathbb{E}[\max(\beta,W_Y^2)],\label{eq_expression0}\\
%\mathbb{E}[(W_{S_{i+1}}-W_{S_i})^4] &= \mathbb{E}[\max(\beta^2,W_Y^4)].\label{eq_expression}
%\end{align}
\end{theorem}

\begin{proof}
See Section \ref{sec_proof}.
\end{proof}

The optimal policy in \eqref{eq_opt_solution} and \eqref{eq_thm1} is called  the ``MMSE-optimal'' policy. 
Note that one can use the bisection method or other one-dimensional search method to solve \eqref{eq_thm1} with quite low complexity. Interestingly, this optimal policy does not suffer from the  ``curse of dimensionality'' issue encountered in many Markov decision problems. 

 {\blue Notice that the feasible policies in $\Pi$ can  use the complete history of the Wiener process $(W_{t}: t\in[0, S_{i+1}])$ to determine $S_{i+1}$. However, the MMSE-optimal policy in \eqref{eq_opt_solution} and \eqref{eq_thm1} only requires recent knowledge of the Wiener process $(W_{t}-W_{S_i}: t\in[S_i+Y_i, S_{i+1}])$ to determine $S_{i+1}$.} 

{\blue Moreover, according to \eqref{eq_thm1}, the threshold $\sqrt{\beta}$ is determined by the maximum sampling frequency $f_{\max}$ and the distribution of the signal variation $W_Y$ during the  channel delay $Y$. It is worth noting that \emph{$W_Y$ is a random variable that tightly couples the source process and the channel delay}. 
%In other words, source and channel are tightly coupled in the MMSE-optimal policy. 
This is different from the traditional wisdom of information theory where source coding and channel coding can be treated separately.}

%, which operates as follows: 

%Hence, this policy is determined by both the signal variations and channel delay. 
%both measurement-aware. 

%Note that $W_{Y}$ and $W_{S_i+Y_i} - W_{S_i}$ are of the same distribution.

%\subsection{Discussions}
%\subsection{Discussions}
\subsection{Signal-Independent Sampling and  the Age-of-Information}
Let $\Pi_{\text{sig-independent}}\subset{\Pi}$ denote the set of signal-independent sampling policies, defined as
\begin{align}
\Pi_{\text{sig-independent}} \!=\! \{\pi\in\Pi: \pi \text{ is independent of }W_t, t\geq 0\}.\nonumber
\end{align}
For each $\pi\in \Pi_{\text{sig-independent}}$, the MMSE  \eqref{eq_DPExpected} can be written as 
\ifreport
(see Appendix \ref{app_age} for its proof)
\fi
{
\begin{align}\label{eq_age_MSE}
%{\mathsf{mmse}}_\pi =
\limsup_{T\rightarrow \infty}\frac{1}{T}\mathbb{E}\left[ \int_0^{T}\Delta (t) dt\right],
\end{align}}
\!\!\!where 
\begin{align}\label{eq_age_def}
\Delta (t) = t - S_i, ~t\in[D_i,D_{i+1}),~i=0,1,2,\ldots,
\end{align}
is the \emph{age-of-information} \cite{KaulYatesGruteser-Infocom2012}, that is, the time difference between the generation time of the freshest received sample and the current time $t$.
%amount of time elapsed since the freshest received sample was generated. 
If the policy space in \eqref{eq_DPExpected}  is restricted from $\Pi$  to $\Pi_{\text{sig-independent}}$, \eqref{eq_DPExpected} reduces to the following age-of-information optimization problem \cite{SunInfocom2016,report_AgeOfInfo2016}:
\begin{align}\label{eq_age}
%\pi_{\text{opt}}=\arg
%{\mathsf{mmse}}_{\text{sig-independent}} = 
\min_{\pi\in\Pi_{\text{sig-independent}} }& \limsup_{T\rightarrow \infty}\frac{1}{T}\mathbb{E}\left[ \int_0^{T}\Delta (t) dt\right] \\
~\text{s.t.}~~~~~ &\liminf_{n\rightarrow \infty} \frac{1}{n} 
\mathbb{E}[S_n]\geq \frac{1}{f_{\max}}.\nonumber
\end{align}

Problem \eqref{eq_DPExpected} is significantly  more challenging than \eqref{eq_age}, because in \eqref{eq_DPExpected} the sampler needs to make decisions based on the evolution of the signal process $W_t$, which is not required in \eqref{eq_age}. 
\ifreport
More powerful techniques 
than those in \cite{SunInfocom2016,report_AgeOfInfo2016} are developed in Section \ref{sec_proof}  to solve \eqref{eq_DPExpected}.
\else
More powerful techniques 
than those in \cite{SunInfocom2016,report_AgeOfInfo2016} are developed in Section \ref{sec_proof} and our  technical report \cite{Sun_reportISIT17}  to solve \eqref{eq_DPExpected}.
\fi
%More powerful techniques 
%than those in \cite{SunInfocom2016,report_AgeOfInfo2016} is developed in our technical report \cite{Sun_reportISIT17} to solve \eqref{eq_DPExpected}.

%the techniques  are insufficient from solving \eqref{eq_DPExpected}.

%
%where ${\mathsf{mmse}}_{\text{sig-independent}}$ is the MSE of the optimal signal-independent policy. 
\begin{theorem}\cite{SunInfocom2016,report_AgeOfInfo2016}\label{thm_2}
There exists $\beta\geq0$ such that the sampling policy \eqref{eq_thm2_1} is optimal to \eqref{eq_age}, and the optimal $\beta$ is determined by solving
%\begin{align}\label{eq_thm2_1}
%S_{i+1}\! =&\! \inf \left\{ t\geq S_i + Y_i:  t -S_i \!\geq \beta\right\}\nonumber\\
%=&S_i + \max(\beta,Y_i),
%\end{align}
%where $\beta$ is determined by 
\begin{align}\label{eq_thm2}
\mathbb{E}[\max(\beta,Y)] \!=\! \max\left(\frac{1}{f_{\max}}, \frac{\mathbb{E}[\max(\beta^2,Y^2)]}{2\beta}\right).\!
\end{align}
%if $\beta>0$ in equation \eqref{eq_thm2}; otherwise, $\beta = 0$.
The optimal  value of \eqref{eq_age} is then given by\emph{
\begin{align}\label{thm_2_obj}
\mathsf{mmse}_{\text{age-opt}}\triangleq\frac{\mathbb{E}[\max(\beta^2,Y^2)]}{2\mathbb{E}[\max(\beta,Y)] } + \mathbb{E}[Y].
\end{align}}
\end{theorem}
%The age-optimal  policy in \eqref{eq_thm2_1} and \eqref{eq_thm2} is signal-independent. 
% and depend only on the communication delay of the channel.
The sampling policy in \eqref{eq_thm2_1} and \eqref{eq_thm2} is referred to as the ``age-optimal'' policy. 
Because $\Pi_{\text{sig-independent}}\subset{\Pi}$, 
\begin{align}\label{eq_compare_aware}
\mathsf{mmse}_{\text{opt}} \leq \mathsf{mmse}_{\text{age-opt}}.
\end{align}
%the optimal objective value of \eqref{eq_DPExpected} is no greater  than that of \eqref{eq_age}. 
%
In the following, the asymptotics of the MMSE-optimal and age-optimal sampling policies at  low/high channel delay or low/high sampling frequencies  are studied. 
\subsection{Low Channel Delay or Low Sampling Frequency}
\ifreport
Let 
\begin{align}\label{eq_scale_Y_i}
Y_i = d X_i
\end{align}
\else
Let $Y_i = d X_i$ 
\fi
represent the scaling of the channel delay $Y_i$ with $d$, where $d\geq0$ and the $X_i$'s are \emph{i.i.d.} positive random variables. If $d\rightarrow  0$ or $f_{\max}\rightarrow 0$, we can obtain from \eqref{eq_thm1} that
\ifreport
(see Appendix \ref{app_low_delay} for its proof)
\fi
\begin{align}\label{eq_betaasy}
\beta = \frac{1}{f_{\max}} + o\left(\frac{1}{f_{\max}}\right),
\end{align}
where $f(x)=o(g(x))$ as $x\rightarrow a$ means that $\lim_{x\rightarrow a} $ $f(x)/g(x) = 0$.
Hence, the MMSE-optimal policy 
\ifreport
in \eqref{eq_opt_solution} and \eqref{eq_thm1} 
\fi
becomes  
\begin{align} \label{eq_opt_solution1}
S_{i+1}\! =\! \inf \left\{ t \geq S_i: |W_t - W_{S_i}| \!\geq\! \sqrt{\frac{1}{f_{\max}}}\right\},
\end{align}
and 
\ifreport 
as shown in Appendix \ref{app_low_delay}, the optimal value of \eqref{eq_DPExpected}  becomes
\begin{align}\label{eq_opt_limit_1}
\mathsf{mmse}_{\text{opt}}= \frac{1}{6f_{\max}} + o\left(\frac{1}{f_{\max}}\right).
\end{align}
\else
the optimal value of \eqref{eq_DPExpected}  becomes
$\mathsf{mmse}_{\text{opt}}= {1}/{(6f_{\max})} $ $+ o(1/f_{\max})$. 
\fi 
The sampling policy \eqref{eq_opt_solution1} 
was also obtained  in \cite{Basar2014} for the case that $ Y_i =  0$ for all $i$.

%
%which coincides with the result in \cite{Ba?ar2014}. 
%there is no communication delay between the sampler and estimator, which was studied . 

If $d\rightarrow  0$ or $f_{\max}\rightarrow 0$, one can show that the age-optimal  policy in \eqref{eq_thm2_1} and \eqref{eq_thm2} becomes  uniform sampling \eqref{eq_uniform} with $\beta = {1}/{f_{\max}}+ o(1/f_{\max})$, and the optimal value of \eqref{eq_age} is $\mathsf{mmse}_{\text{age-opt}}= {1}/{(2f_{\max})}+ o(1/f_{\max})$. Therefore, 
\begin{align}\label{eq_ratio_MSE}
\lim_{d\rightarrow  0} \frac{\mathsf{mmse}_{\text{opt}}}{\mathsf{mmse}_{\text{age-opt}}}= \lim_{f_{\max}\rightarrow 0}\frac{\mathsf{mmse}_{\text{opt}}}{\mathsf{mmse}_{\text{age-opt}}}=\frac{1}{3}.
\end{align}

\subsection{High Channel Delay or Unbounded Sampling Frequency}
If $d\rightarrow\infty$ or $f_{\max}\rightarrow\infty$, 
\ifreport
as shown in Appendix \ref{app_scale}, 
\fi
the MMSE-optimal policy for solving \eqref{eq_DPExpected}
is given by \eqref{eq_opt_solution} where $\beta$ is determined by solving 
\begin{align} \label{eq_coro_1}
2\beta \mathbb{E}[\max(\beta,W_{Y}^2)]= {\mathbb{E}[\max(\beta^2,W_{Y}^4)]}{}.
\end{align}
%if $\beta>0$ in equation \eqref{eq_coro_1}; otherwise, $\beta = 0$. 
%if $\beta>0$ in equation \eqref{eq_coro_1}; otherwise, $\beta = 0$. On the other hand, if $Y_i$ scales up by $a$ times, then $\beta$ and $6\mathbb{E}[\max(\beta,W_{Y}^2)]$ in \eqref{eq_coro_1} also scale up by $a$ times. Hence, if $d\rightarrow\infty$ in \eqref{eq_scale_Y_i}, the MMSE-optimal sampling policy is  given by \eqref{eq_opt_solution} and \eqref{eq_coro_1}.
Similarly, if $d\rightarrow\infty$ or $f_{\max}\rightarrow\infty$, the age-optimal policy for solving \eqref{eq_age} is given by \eqref{eq_thm2_1} where $\beta$ is determined by solving 
\begin{align} \label{eq_coro_2}
2\beta \mathbb{E}[\max(\beta,Y)] = {\mathbb{E}[\max(\beta^2,Y^2)]}{}.
\end{align}
%if $\beta>0$ in equation \eqref{eq_coro_2}; otherwise, $\beta = 0$. In this case, 
In these limits, the ratio between ${\mathsf{mmse}_{\text{opt}}}$ and ${\mathsf{mmse}_{\text{age-opt}}}$ depends on the distribution of $Y$.

When the sampling frequency is unbounded, i.e., $f_{\max}=\infty$, one logically reasonable policy is the zero-wait policy  in \eqref{eq_Zero_wait} \cite{2015ISITYates,SunInfocom2016,report_AgeOfInfo2016,KaulYatesGruteser-Infocom2012}. This zero-wait policy achieves the maximum throughput and the minimum queueing delay of the channel. Surprisingly, this zero-wait policy  \emph{does not always} minimize the age-of-information in \eqref{eq_age} and \emph{almost never} minimizes the MMSE in \eqref{eq_DPExpected}, as stated below:
%in the following two theorems:

\begin{theorem}\label{lem_zero_wait1}
If $f_{\max}=\infty$, the zero-wait policy is optimal for solving \eqref{eq_DPExpected} if and only if $Y=0$ with probability one.
\end{theorem}
\ifreport
\begin{proof}
See Appendix \ref{app_zerowait}. 
\end{proof}
\fi
\begin{theorem}\label{lem_zero_wait2}\cite{report_AgeOfInfo2016}
If $f_{\max}=\infty$, the zero-wait policy is optimal for solving \eqref{eq_age} if and only if \emph{
\begin{align}\label{eq_zero_wait2}
\mathbb{E}[Y^2]\leq 2~ \text{ess}\inf Y~ \mathbb{E}[Y],
\end{align}
where $\text{ess}\inf Y = \sup\{y\in[0,\infty): \Pr[Y< y]=0\}$.}
\end{theorem}
\ifreport
\begin{proof} 
See Appendix \ref{app_zerowait}. 
\end{proof}
Note that the condition in Theorem \ref{lem_zero_wait2} is weaker than that of Theorem 5 in \cite{report_AgeOfInfo2016}.
\else
Theorems \ref{lem_zero_wait1} and \ref{lem_zero_wait2} are proven in our technical report \cite{Sun_reportISIT17}.
 \fi

% !TEX root = ./sampling_Wiener.tex

%In this section, we develop a simple optimal sampling policy for solving \eqref{eq_DPExpected}. % by using  martingale, Lagrangian duality, and  optimal stopping theories. % \cite{Shiryaev1978}.

%propose a simple sampling policy and show that it is optimal for 

%Somewhat surprisingly, we find that it is possible to obtain a solution that is optimal for solving \eqref{eq_DPExpected}. % has some nice structural properties that can be exploited to obtain the optimal sample policy. 
\section{Proof of the Main Result}\label{sec_proof}

We establish Theorem \ref{thm_1} in four steps: First, we employ the strong Markov property  of the Wiener process to simplify the optimal sampling problem \eqref{eq_DPExpected}. Second, we study the Lagrangian dual problem of the simplified problem, and decompose the Lagrangian dual problem into a series of \emph{mutually independent} per-sample control problems. Each of these per-sample control problems is a continuous-time Markov decision problem. Third, we utilize optimal stopping theory %\cite{Shiryaev1978,Peskir2006} 
\cite{Shiryaev1978} 
to solve the per-sample control problems. Finally, we show that the Lagrangian duality gap is zero. By this, the original problem \eqref{eq_DPExpected} is solved. 
The details are as follows.
% in the sequel.
%In this section, we will solve \eqref{eq_DPExpected} by using 
\subsection{Problem Simplification} % 
%\begin{lemma}
%Let $S\leq T$ be two stopping times of the Wiener process $W_t$ and $\mathbb{E}[T]<\infty$. Then
%\begin{align}
%\mathbb{E}\left[ W_t^4 \right] = \mathbb{E}\left[ W_s^4 \right] + \mathbb{E}\left[ (W_t- W_s)^4 \right].  
%\end{align}
%\end{lemma}
%If a sample is stored in the queue, it becomes stale while waiting for its transmission opportunity. 
%A better method is to wait until the channel becomes idle, and then generate a new sample, which is stated in the following lemma. 

We first provide a lemma that is crucial for simplifying \eqref{eq_DPExpected}.
\begin{lemma}\label{lem_zeroqueue}
%In \eqref{eq_DPExpected},  
In the optimal sampling problem \eqref{eq_DPExpected} for minimizing
the MMSE of the Wiener process, it is suboptimal to take a new sample before the previous sample is delivered.

%it is better to take each new sample and  after  all previous samples are delivered.

%it is better to wait until the channel becomes idle, and then generate each new sample and immediately send it out.

%it is better to take each new sample and  after  all previous samples are delivered.

%
\end{lemma}
\begin{proof}
See Appendix \ref{app_zeroqueue}. 
\end{proof}
{\blue In recent studies on age-of-information  \cite{SunInfocom2016,report_AgeOfInfo2016}, Lemma \ref{lem_zeroqueue} was intuitive and hence was used without a proof: If a sample is taken when the channel is busy, it needs to wait in the queue until its transmission starts, and becomes stale while waiting. A better method is 
to wait until the channel becomes idle, and then generate a new
sample, as stated in Lemma \ref{lem_zeroqueue}. However, this lemma is not intuitive in the MMSE minimization problem \eqref{eq_DPExpected}: The proof of Lemma \ref{lem_zeroqueue} relies on the strong Markov property of Wiener process, which may not hold for other signal processes. 
}

By Lemma \ref{lem_zeroqueue}, we only need to consider a sub-class of sampling policies $\Pi_1\subset\Pi$ such that each  sample is generated and submitted to the channel after the previous sample is delivered, i.e., 
\begin{align}
\Pi_1 %&= \{\pi\in\Pi: S_{i} = G_{i} \geq D_{j} \text{ for all $i$ and $j=1,\ldots,i-1$}\} \nonumber \\
 &= \{\pi\in\Pi: S_{i} = G_{i} \geq D_{i-1} \text{ for all $i$}\}. 
\end{align}
%In particular, sample $i$ is delivered before sample $i+1$ is generated and sent out, i.e., $D_{i} \leq S_{i+1} = G_{i+1}$.
This completely eliminates the waiting time wasted in the queue, and hence the queue is always kept empty. 
The \emph{information} that is available for determining $S_i$ includes the history of signal values $(W_t: t\in[0, S_i])$ and the channel delay  $(Y_1,\ldots, Y_{i-1})$ of previous samples.\footnote{Note that the generation times $(S_1,\ldots, S_{i-1})$ of previous samples are also included in this information.} To characterize this statement precisely, let us define the $\sigma$-fields $\mathcal{F}_t = \sigma(W_s: s\in[0, t])$ and $\mathcal{F}_t^+ = \cap_{s>t}\mathcal{F}_s$. Then, $\{\mathcal{F}_t^+,t\geq0\}$ is the {filtration} (i.e., a non-decreasing and right-continuous  family of  $\sigma$-fields) of the Wiener process $W_t$.
%, and $\{\mathcal{F}_C (i), i = 1,2,\ldots\}$ is the \emph{filtration} of the channel transmission process. Then, 
Given the transmission  durations $(Y_1,\ldots, Y_{i-1})$ of previous samples, $S_{i} $ is a \emph{stopping time} with respect to the filtration $\{\mathcal{F}_t^+,t\geq0\}$ of the Wiener process $W_t$, that is
\begin{align}
[\{S_{i}\leq t\} | Y_1,\ldots, Y_{i-1}] \in \mathcal{F}_t^+.\label{eq_stopping}
\end{align}  
Then, the policy space $\Pi_1$ can be alternatively expressed as
\begin{align}\label{eq_policyspace}
\!\!\!\!\Pi_1 = & \{S_i : [\{S_{i}\leq t\} | Y_1,\ldots, Y_{i-1}] \in \mathcal{F}_t^+, \nonumber\\
&~~~~~~~S_{i} = G_{i} \geq D_{i-1} \text{ for all $i$}, \nonumber\\
&~~~~~~~\text{$T_i = S_{i+1}-S_i$ is a regenerative process} \}.\!\!
\end{align}  
%This %interpretation 
% which will be used to simplify \eqref{eq_DPExpected}:
% for simplifying \eqref{eq_DPExpected}. %  which we obtain  the following result:
%This interpretation motivated us to 

%From \eqref{eq_esti}, we can get
%\begin{align}%\label{eq_integral}
%\mathbb{E}\!\left\{
%\!\int_{D_{i}}^{D_{i+1}} \!\![W_t\!-\!\hat W_t]^2dt\right\}\!{=} \mathbb{E}\!\left\{\!\int_{D_{i}}^{D_{i+1}} \! \![W_t\!-\!B(S_{i})]^2dt\right\}\!.\! \nonumber
%\end{align}

Let $Z_i = S_{i+1} - D_{i}\geq0$ represent the \emph{waiting time} between the delivery time $D_{i}$ of sample $i$  and the generation time $S_{i+1}$ of sample $i+1$. Then,
$S_i =  Z_{0} +\sum_{j=1}^{i-1} (Y_{j} + Z_j)$ and $D_i =  \sum_{j=0}^{i-1} (Z_{j}+Y_{j+1})$. If $(Y_1,Y_2,\ldots)$ is given, $(S_0,S_1,\ldots)$ is uniquely determined by $(Z_0,Z_1,\ldots)$. Hence, one can also use $\pi = (Z_0,Z_1,\ldots)$ to represent a sampling policy. 

Because $T_i $ is a regenerative process, by following the renewal theory in \cite{Ross1996} and \cite[Section 6.1]{Haas2002}, one can show that $\frac{1}{n} 
\mathbb{E}[S_n]$ is a convergent sequence and %$(Z_1,Z_2,\ldots)$ is also stationary and ergodic. Hence, 
%Because  $(Y,Y_2,\ldots)$ are \emph{i.i.d.} and $(Z_1,Z_2,\ldots)$ are stationary and ergodic, we can obtain
\begin{align}
&\limsup_{T\rightarrow \infty}\frac{1}{T}\mathbb{E}\left[\int_0^{T} (W_t-\hat W_t)^2dt\right] \nonumber
\\
=& \lim_{n\rightarrow \infty}\frac{\mathbb{E}\left[\int_0^{D_n} (W_t-\hat W_t)^2dt\right]}{\mathbb{E}[D_n]} \nonumber\\
=& \lim_{n\rightarrow \infty}\frac{\sum_{i=0}^{n-1}\mathbb{E}\left[\int_{D_{i}}^{D_{i+1}} (W_t- W_{S_{i}})^2dt\right]}{\sum_{i=0}^{n-1} \mathbb{E}\left[Y_i+Z_i\right]},\nonumber
%= & \lim_{n\rightarrow \infty}\frac{\sum_{i=1}^n\mathbb{E}\left[(W_{S_{i}+Y_i+Z_i}-W_{S_i})^4\right]}{ 6\sum_{i=1}^n \mathbb{E}\left[Y_i+Z_i\right]} +  \mathbb{E}\left[Y\right],\nonumber
\end{align}
where in the last step we have used $\mathbb{E}\left[D_n\right]=\mathbb{E}[\sum_{i=0}^{n-1} (Z_{i}+Y_{i+1})] = \mathbb{E}[\sum_{i=0}^{n-1} (Y_{i}+Z_{i})]$.
Hence, \eqref{eq_DPExpected} can be rewritten as the following Markov decision problem:
\begin{align}\label{eq_Simple}
{\mathsf{mmse}}_{\text{opt}}\triangleq&\min_{\pi\in\Pi_1} \lim_{n\rightarrow \infty}\frac{\sum_{i=0}^{n-1}\mathbb{E}\left[\int_{D_{i}}^{D_{i+1}} (W_t\!-\!W_{S_{i}})^2dt\right]}{\sum_{i=0}^{n-1} \mathbb{E}\left[Y_i+Z_i\right]} \\
&~\text{s.t.}~\lim_{n\rightarrow \infty} \frac{1}{n} 
\sum_{i=0}^{n-1} \mathbb{E}\left[Y_i+Z_i\right]\geq \frac{1}{f_{\max}},\label{eq_Simple_constraint}
\end{align}
where ${\mathsf{mmse}}_{\text{opt}}$ is the optimal  value of \eqref{eq_Simple}.

In order to solve \eqref{eq_Simple}, let us consider the following Markov decision problem with a parameter $c\geq 0$:
%In order to solve \eqref{eq_Simple}, let us consider the following Markov decision with a parameter $c$:
\begin{align}\label{eq_SD}
\!\!p(c)\!\triangleq\!\min_{\pi\in\Pi_1}&\lim_{n\rightarrow \infty}\frac{1}{n}\sum_{i=0}^{n-1}\!\mathbb{E}\!\left[\int_{D_{i}}^{D_{i+1}} \!\!\!\!(W_t-W_{S_{i}})^2dt\!-\! c(Y_i\!+\!Z_i)\!\right]\!\!\!\!\\
\text{s.t.}~&\lim_{n\rightarrow \infty} \frac{1}{n} 
\sum_{i=0}^{n-1} \mathbb{E}\left[Y_i+Z_i\right]\geq \frac{1}{f_{\max}},\nonumber
\end{align}
where $p(c)$ is the optimum  value of \eqref{eq_SD}. 
\begin{lemma} \label{lem_ratio_to_minus}
The following assertions are true:
\begin{itemize}
\vspace{0.5em}
\item[(a).] \emph{${\mathsf{mmse}}_{\text{opt}} \gtreqqless c $} if and only if $p(c)\gtreqqless 0$. 
\vspace{0.5em}
\item[(b).] If $p(c)=0$, the solutions to \eqref{eq_Simple}
and \eqref{eq_SD} are identical. 
\end{itemize}
\end{lemma}
\begin{proof}
See Appendix \ref{app_ratio_to_minus}.
\end{proof}
Hence, the solution to \eqref{eq_Simple} can be obtained by solving \eqref{eq_SD} and seeking a $c_{\text{opt}}\geq 0$ such that 
\begin{align}\label{eq_c}
p(c_{\text{opt}})=0. 
\end{align}

\subsection{Lagrangian Dual Problem of  \eqref{eq_SD}}

Although \eqref{eq_SD} is a continuous-time Markov decision problem with a continuous state space (not a convex optimization problem), it is possible to use the Lagrangian dual approach to solve \eqref{eq_SD} and show that it admits no duality gap. 

Define the following Lagrangian
\begin{align}\label{eq_Lagrangian}
& L(\pi;\lambda,c) \nonumber\\
=& \lim_{n\rightarrow \infty}\frac{1}{n}\sum_{i=0}^{n-1}\mathbb{E}\left[\int_{D_{i}}^{D_{i+1}} (W_t-W_{S_{i}})^2dt\!-\! (c+\lambda)(Y_i\!+\!Z_i)\right] \nonumber\\
&~+ \frac{\lambda}{f_{\max}}.
\end{align}
Let
\begin{align}\label{eq_primal}
g(\lambda,c) \triangleq \inf_{\pi\in\Pi_1}  L(\pi;\lambda,c).
\end{align}
Then, the Lagrangian dual problem of \eqref{eq_SD} is defined by
\begin{align}\label{eq_dual}
d(c) \triangleq \max_{\lambda\geq 0}g(\lambda,c),
\end{align}
where $d(c)$ is the optimum value of \eqref{eq_dual}. 
Weak duality \cite{Bertsekas2003,Boyd04} implies that $d(c) \leq p(c)$. 
In Section \ref{sec:dual}, we will establish strong duality, i.e., $d(c_{\text{opt}}) = p(c_{\text{opt}}) = 0$, 
at the optimal choice of $c = c_{\text{opt}}$.

%The optimal values of \eqref{eq_SD} and \eqref{eq_dual} are equal

% Strong duality between Problem \eqref{eq_SD} and Problem \eqref{eq_dual} will be established later in .

In the sequel, we solve \eqref{eq_primal}.
% and \eqref{eq_dual} at $c = c_{\text{opt}}$,  where $c_{\text{opt}}$ is defined in \eqref{eq_c}.
By considering the sufficient statistics of the Markov decision problem \eqref{eq_primal}, we obtain
\begin{lemma}\label{lem_decompose}
For any $\lambda \geq 0$, there exists an optimal solution $(Z_0,Z_1,\ldots)$ to \eqref{eq_primal} in which $Z_i$ is independent of $(W_t, t\in[0,W_{S_i}])$ for all $i=1,2,\ldots$
%  satisfies
%\begin{align}\label{eq_lem_decompose}
%=&  \min_{\tau  \in \mathfrak{M}_{S_i+Y_i}} \mathbb{E}\left[\int_{S_{i}+Y_i}^{S_{i}+Y_i+\tau +Y_{i+1}} (W_t-W_{S_{i}})^2dt\right.\nonumber\\
%&~~~~~~~~~\left.- (c+\lambda)(Y_i+\tau)\bigg|W_{S_{i}+Y_i}-W_{S_i},Y_i \right].
%\end{align}
\end{lemma}
\begin{proof}
See Appendix \ref{app_decompose}.
\end{proof}

Using the stopping times and martingale theory of the Wiener process, we obtain the following  lemma: 
%
%that is useful for simplifying \eqref{eq_DPExpected}: %that is obtained along this direction 
%is the following lemma:
\begin{lemma}\label{lem_stop}
Let $\tau \geq0$ be a stopping time of the Wiener process $W_t$ with $\mathbb{E}[\tau^2]<\infty$, 
then
\begin{align}\label{eq_stop}
%\mathbb{E}\left[\int_0^T [{x}(0)-W_t]^2dt\right] = \mathbb{E}[T] \sigma^2.
\mathbb{E}\left[\int_0^\tau W_t^2 dt\right]= \frac{1}{6}\mathbb{E}\left[ W_\tau^4 \right].
\end{align}
\end{lemma}
\begin{proof}
See Appendix \ref{applem_stop}.
\end{proof}

If $Z_i$ is independent of $(W_t, t\in[0,W_{S_i}])$, by using Lemma \ref{lem_stop},  we can show that for every $i = 1,2,\ldots$,
\begin{align}\label{eq_integral}
&\mathbb{E}\left[\int_{D_{i}}^{D_{i+1}} (W_t- W_{S_i})^2dt\right] \nonumber\\
%=& \mathbb{E}\left\{\int_{D_{i}}^{D_{i+1}} [W_t-B(S_i)]^2dt\right\} \nonumber\\
%= & \mathbb{E}\left\{\int_{Y_i}^{Y_i+Z_i+Y_{i+1}} W_t^2dt\right\} \nonumber\\
%= & \frac{1}{6}\mathbb{E}\left[ B(Y_i+Z_i+Y_{i+1})^4 \right] - \frac{1}{6}\mathbb{E}\left[ B(Y_i)^4 \right] \nonumber\\
= & \frac{1}{6}\mathbb{E}\left[ (W_{S_{i}+Y_i+Z_i}-W_{S_i})^4 \right] + \mathbb{E}\left[Y_i + Z_i \right] \mathbb{E}\left[Y_i \right],  
\end{align}
which is proven in Appendix \ref{app_integral}.

Define the $\sigma$-fields $\mathcal{F}^{s}_t = \sigma(W_{s+v}-W_{s}: v\in[0, t])$ and $\mathcal{F}^{s+}_t = \cap_{v>t} \mathcal{F}^{s}_v$, as well as the {filtration} $\{\mathcal{F}^{s+}_t,t\geq0\}$ of the {time-shifted Wiener process} $\{W_{s+t}-W_{s},t\in[0,\infty)\}$. Define ${\mathfrak{M}}_s$ as the set of square-integrable stopping times of $\{W_{s+t}-W_{s},t\in[0,\infty)\}$, i.e.,
\begin{align}
\mathfrak{M}_{s} = \{\tau \geq 0:  \{\tau\leq t\} \in \mathcal{F}^{s+}_t, \mathbb{E}\left[\tau^2\right]<\infty\}.\nonumber
\end{align}
By using %Lemmas \ref{lem_decompose}-\ref{lem_stop} %, %Theorem \ref{thm_zero_gap},
  \eqref{eq_integral} and considering the sufficient statistics of the Markov decision problem \eqref{eq_primal},  we  obtain
\begin{theorem}\label{thm_solution_form}
%If $c = c_{\text{opt}}$, a
An optimal solution  $(Z_0,Z_1,\ldots)$ %and $\lambda$ 
to %\eqref{eq_Simple}, \eqref{eq_SD}, and equivalently 
\eqref{eq_primal} %-\eqref{eq_dual} 
satisfies 
\begin{align}\label{eq_opt_stopping}
Z_i =&  \min_{\tau  \in \mathfrak{M}_{S_i+Y_i}} \mathbb{E}\left[\frac{1}{2}(W_{S_{i}+Y_i+\tau}-W_{S_i})^4\right.\nonumber\\
&~~~~~~~~~~~~~~~~\left.- {\beta}(Y_i+\tau)\bigg|W_{S_{i}+Y_i}-W_{S_i},Y_i \right],
\end{align}
where $\beta$ is given by 
\begin{align}\label{eq_beta_new}
\beta = 3(c + \lambda - \mathbb{E}\left[Y \right]).
\end{align}
\ignore{where $\beta \triangleq 3(c_{\text{opt}} + \lambda - \mathbb{E}\left[Y \right]) \geq0$ is determined by solving
\begin{align}\label{eq_equation}
\mathbb{E}[Y_i+Z_i] \!=\! \max\left(\frac{1}{f_{\max}}, \frac{\mathbb{E}[(W_{S_{i}+Y_i+Z_i}-W_{S_i})^4]}{2\beta}\right).\!
\end{align}
The optimal  value of \eqref{eq_Simple} is then given by \emph{
\begin{align}\label{thm_3_obj}
\mathsf{mmse}_{\text{opt}}=\frac{\mathbb{E}[(W_{S_{i}+Y_i+Z_i}-W_{S_i})^4]}{6\mathbb{E}[Y_i+Z_i] } + \mathbb{E}[Y].
\end{align}}}
\end{theorem}
\begin{proof}
See Appendix \ref{app_solution_form}. 
\end{proof}

\subsection{Per-Sample Optimal Stopping Solution to \eqref{eq_opt_stopping}}
We use  optimal stopping theory  \cite{Shiryaev1978} 
%\cite{Shiryaev1978,Peskir2006}
to solve \eqref{eq_opt_stopping}.
Let us first pose \eqref{eq_opt_stopping} in the language of  optimal stopping.
A continuous-time two-dimensional Markov chain $X_t$ on a probability space $(\mathbb{R}^2,\mathcal{F},\mathbb{P})$ is defined as follows:  Given the initial state $ X_0=x = (s,b)$, the state $X_t$ at time $t$ is  
\begin{align}\label{eq_Markov}
X_t= (s + t, b + W_t), 
\end{align}
where $\{W_t,t\geq 0\}$ is a standard Wiener process. Define $\mathbb{P}_x(A) = \mathbb{P}(A | X_0 = x)$ and $\mathbb{E}_x Z = \mathbb{E}(Z | X_0 = x)$, respectively, 
as the conditional probability of event $A$ and the conditional expectation of random variable $Z$ 
for given  initial state $X_0 =x$. Define the $\sigma$-fields $\mathcal{F}^{X}_t  = \sigma(X_{v}: v\in[0, t])$ and $\mathcal{F}^{X+}_t  = \cap_{v>t} \mathcal{F}^{X}_v$, as well as the {filtration} $\{\mathcal{F}^{X+}_t ,t\geq0\}$ of the Markov chain $X_t$.
A random variable $\tau: \mathbb{R}^2\rightarrow [0,\infty)$ is said to be a \emph{stopping time} of $X_t$ if $\{\tau \leq t\}\in \mathcal{F}^{X+}_t $ {for all}~$t\geq0$. Let ${\mathfrak{M}}$ be the set of  square-integrable stopping times of $X_t$, i.e.,
\begin{align}
\mathfrak{M} = \{\tau \geq 0:  \{\tau\leq t\} \in \mathcal{F}^{X+}_t , \mathbb{E}\left[\tau^2\right]<\infty\}.\nonumber
\end{align}

%Let $\{\mathcal{F}(t),t\geq0\}$ be a \emph{filtration} of the process $\{X(t),t\geq 0\}$, that is, a family of $\sigma$-fields 
%$\mathcal{F}(t) = \sigma(X(s):0\leq s \leq t)$ such that $\mathcal{F}(s) \subset \mathcal{F}(t)  \subset \mathcal{F}$ for all $0\leq s< t$. 
%A random variable $\tau: \mathbb{R}^2\rightarrow [0,\infty]$ is said to be a \emph{stopping time} with respect to the filtration $\{\mathcal{F}(t),t\geq0\}$ if $\{\tau \leq t\}\in \mathcal{F}(t)$ {for all}~$t\geq0$.
%Let $\overline{\mathfrak{M}}$ be the set of all stopping times and $\mathfrak{M}$ be the set of all finite stopping times, i.e.,\begin{align}
%\mathfrak{M} = \{\tau\in \overline{\mathfrak{M}}: \mathbb{P}_x(\tau<\infty) = 1,~\forall~x\in \mathbb{R}^2\}.
%\end{align}
Our goal is to solve the following optimal stopping problem:
\begin{align}\label{eq_stop_prob}
\sup_{\tau\in \mathfrak{M}}  \mathbb{E}_x g(X_\tau),
\end{align}
where the function $g: \mathbb{R}^2\rightarrow \mathbb{R}$ is defined as 
\begin{align}\label{eq_fun}
g(s,b) =\beta s  - \frac{1}{2}b^4
\end{align}
with parameter $\beta \geq 0$, and $x$ is the initial state of the Markov chain $X(t)$. Notice that \eqref{eq_stop_prob} becomes  \eqref{eq_opt_stopping}
if the initial state is $ x=(Y_i,W_{S_{i}+Y_i}-W_{S_i})$,  and  $W_t$ is replaced by $W_{S_{i}+Y_i+t}-W_{S_i}$. 
\begin{figure}
\centering \includegraphics[width=0.35\textwidth]{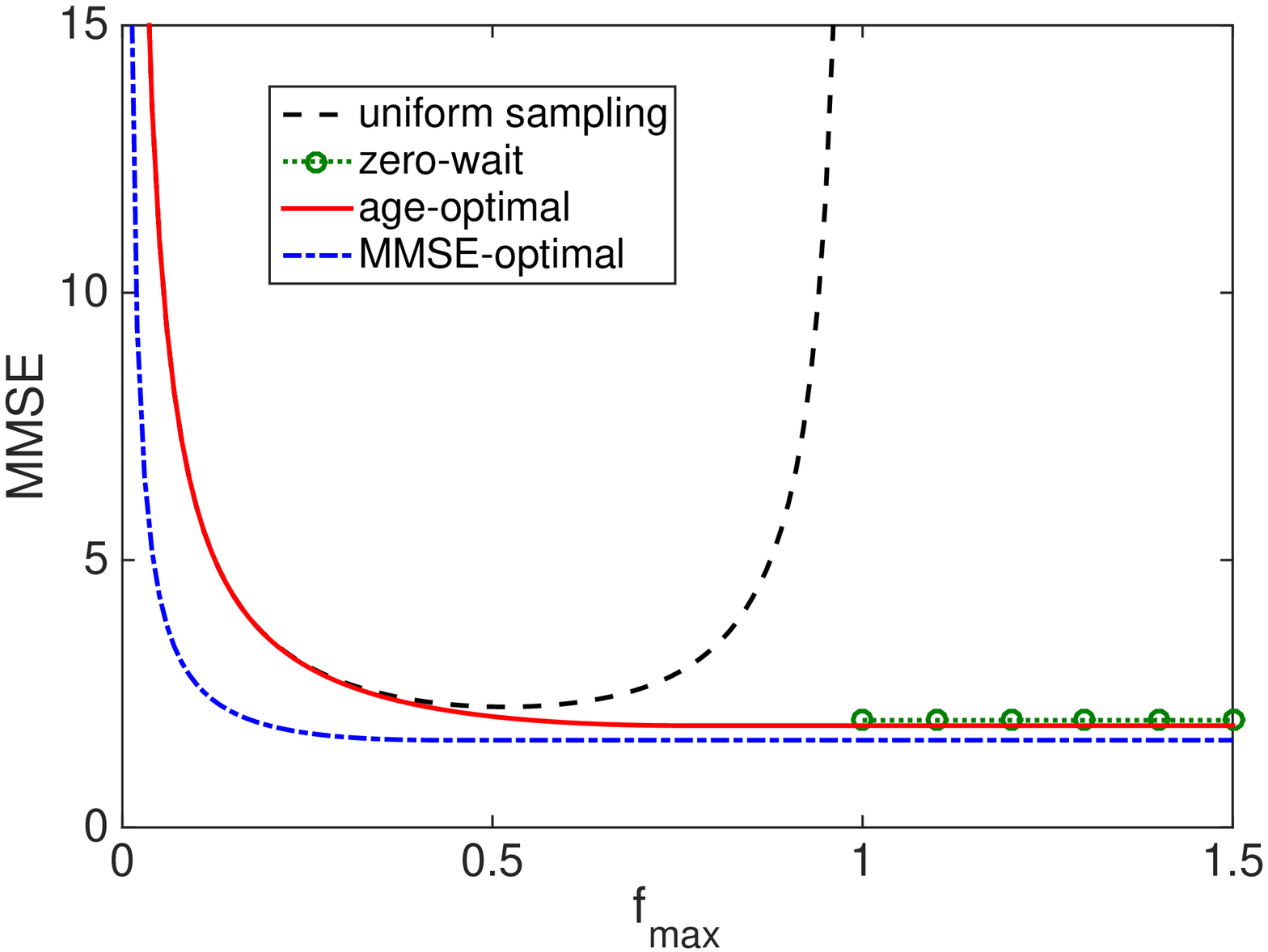} \caption{MMSE vs. ${f_{\max}}$ tradeoff for \emph{i.i.d.} exponential channel delay.}
\label{fig1} \vspace{-0.cm}
\end{figure}

\begin{theorem}\label{thm_optimal_stopping}
For all initial state $(s,b)\in \mathbb{R}^2$ and  $\beta\geq0$, an optimal stopping time for solving \eqref{eq_stop_prob} is
\begin{align}\label{eq_opt_stop_solution}
\tau^* = \inf \left\{ t \geq 0:\! \left|b+W_t\right| \!\geq\! \sqrt{\beta}\right\}.
\end{align}
\end{theorem} 

%We will prove Theorem \ref{thm_optimal_stopping} by using the optimal stopping theory of Markov chain \cite{Shiryaev1978,Peskir2006}. To that end, 

In order to prove Theorem \ref{thm_optimal_stopping}, let us define the function
$u(x) = \mathbb{E}_x g(X_{\tau^*})$ and establish some properties of $u(x)$.

\begin{lemma} \label{lem1_stop}
 $u(x)\geq g(x)$ for all $x\in \mathbb{R}^2$, and 
\begin{align}
u(s,b) = \left\{\begin{array}{l l} 
\beta s- \frac{1}{2}b^4 , &\text{if}~ b^2 \geq \beta;
\vspace{0.5em}\\
\beta s+ \frac{1}{2} \beta ^2- \beta b^2 , &\text{if}~ b^2 < \beta.
\end{array}\right.
\end{align}
\end{lemma} 
\begin{proof}
See Appendix \ref{app_optimal_stopping}. 
\end{proof}

A function $f(x)$ is said to be \emph{excessive} for the process $X_t$ if \cite{Shiryaev1978}
\begin{align}
\mathbb{E}_x f(X_t)\leq f(x),~\text{for all}~ t\geq 0, ~x\in \mathbb{R}^2.
\end{align} 
By using the It\^{o}-Tanaka-Meyer formula in stochastic calculus, we can obtain
\begin{lemma}\label{lem2_stop}
The function $u(x)$ is excessive for the process $X_t$.
\end{lemma} 
%\begin{figure}
%\centering \includegraphics[width=0.35\textwidth]{./figure1} \caption{MMSE vs. ${f_{\max}}$ tradeoff for \emph{i.i.d.} exponential channel delay.}
%\label{fig1} \vspace{-0.cm}
%\end{figure}

\begin{proof}
See Appendix \ref{app_optimal_stopping1}. 
\end{proof}

Now, we are ready to prove Theorem \ref{thm_optimal_stopping}. 
\begin{proof}[Proof of Theorem \ref{thm_optimal_stopping}]
In Lemma \ref{lem1_stop} and Lemma \ref{lem2_stop}, we have shown that $u(x)=\mathbb{E}_x g(X_{\tau^*})$ is an excessive function and $u(x)\geq g(x)$. In addition, it is known that $\mathbb{P}_x(\tau^*<\infty) = 1$ for all $x\in \mathbb{R}^2$ \cite[Theorem 8.5.3]{Durrettbook10}. These conditions, together with the Corollary to Theorem 1 in \cite[Section 3.3.1]{Shiryaev1978}, imply that $\tau^*$ is an optimal stopping time of \eqref{eq_stop_prob}. This completes the proof.
\end{proof}

An immediate consequence of Theorem \ref{thm_optimal_stopping} is
\begin{corollary}\label{coro_stop}
An optimal solution to \eqref{eq_opt_stopping} is
\begin{align}\label{eq_opt_stop_solution1}
Z_i = \left\{\!\!\!\begin{array}{c c}\inf \left\{ t \geq 0:\! \left|W_{S_i+Y_i+t}-W_{S_i}\right| \geq \sqrt{\beta}\right\},& \!\!\text{if } \beta \geq 0 ;\\ 0, & \!\!\text{if } \beta < 0.\!\!
\end{array}\right.
\end{align}
%if $\beta<0$, then an optimal solution to \eqref{eq_opt_stopping} is $Z_i = 0$.
%and the expectations in \eqref{eq_equation} %and \eqref{thm_3_obj} 
%are then given by
%\begin{align}
%&\mathbb{E}[Y_i+Z_i] = \mathbb{E}[\max(\beta,W_Y^2)],\label{eq_expression0}\\
%&\mathbb{E}[(W_{S_{i}+Y_i+Z_i}-W_{S_i})^4] = \mathbb{E}[\max(\beta^2,W_Y^4)].\label{eq_expression}
%\end{align}
\end{corollary}
%\begin{proof}
%See Appendix \ref{app_expression}.
%\end{proof}
%Notice that \eqref{eq_opt_solution} is equivalent to \eqref{eq_opt_stop_solution1}. 

%According to Theorem \ref{thm_solution_form} and Corollary \ref{coro_stop}, the MMSE-optimal  policy in \eqref{eq_opt_solution} and \eqref{eq_thm1} achieves the optimal value of \eqref{eq_dual} at $c = c_{\text{opt}}$. 

%Finally, \eqref{eq_expression0} and \eqref{eq_expression} are proven in Appendix \ref{app_expression}. 
%This completes the proof of Theorem \ref{thm_1}.

\subsection{Zero Duality Gap between \eqref{eq_SD} and \eqref{eq_dual} at $c = c_{\text{opt}}$ } \label{sec:dual}
%Strong duality is established in the following theorem:
\begin{theorem}\label{thm_zero_gap}
The following assertions are true:
\begin{itemize}
\vspace{0.5em}
\item[(a).] The duality gap between \eqref{eq_SD} and \eqref{eq_dual} is zero such that $d(c) = p(c)$.
\vspace{0.5em}
\item[(b).] A common optimal solution to \eqref{eq_DPExpected},  \eqref{eq_Simple}, and \eqref{eq_SD} with \emph{$c = c_{\text{opt}}$},  is given by \eqref{eq_opt_solution} and \eqref{eq_thm1}. 
\end{itemize}
%Hence, the MMSE-optimal  policy in \eqref{eq_opt_solution} and \eqref{eq_thm1} is an optimal solution to \eqref{eq_SD}. 
\end{theorem}

\begin{proof}[Proof Sketch of Theorem \ref{thm_zero_gap}]
We first use Theorem \ref{thm_solution_form} and Corollary \ref{coro_stop} to find a geometric multiplier \cite{Bertsekas2003} for the primal problem \eqref{eq_SD}. Hence, the duality gap between \eqref{eq_SD} and \eqref{eq_dual} is zero, because otherwise there exists no geometric multiplier \cite[Section 6.1-6.2]{Bertsekas2003}. By this, part (a) is proven. Next, we consider the case of $c =  c_{\text{opt}}$ and use Lemma \ref{lem_ratio_to_minus}  to prove part (b).
%Let $\pi^{\star}=(Z_0^\star, Z_1^\star,\ldots)$ be the sampling policy in \eqref{eq_opt_solution} and \eqref{eq_thm1}, and let $\lambda^\star$ be the optimal dual solution to \eqref{eq_dual} with {$c = c_{\text{opt}}$}. We  show that $\pi^{\star}$ and $\lambda^\star$ satisfies the following conditions:
%\begin{align}
%&\pi^{\star}\in\Pi, \lim_{n\rightarrow \infty} \frac{1}{n} \sum_{i=0}^{n-1} \mathbb{E}\left[Y_i+Z_i^\star\right] - \frac{1}{f_{\max}} \geq  0,\label{eq_mix_or_not0}\\
%&\lambda^{\star}\geq 0,\label{eq_mix_or_not1}\\
%&L(\pi^{\star};\lambda^{\star},c_{\text{opt}}) = \inf_{\pi\in\Pi_1}  L(\pi;\lambda^{\star},c_{\text{opt}}),\label{eq_mix_or_not}\\
%&\lambda^\star \left\{\lim_{n\rightarrow \infty} \frac{1}{n} \sum_{i=0}^{n-1} \mathbb{E}\left[Y_i+Z_i^\star\right] - \frac{1}{f_{\max}}\right\} = 0.\label{eq_KKT_last}
%\end{align}
%According to \cite[Prop. 6.2.5]{Bertsekas2003}, $\pi^{\star}$ is an optimal solution to the primal problem \eqref{eq_SD} and $\lambda^\star$ is a geometric multiplier with  $c = c_{\text{opt}}$. By Lemma \ref{lem_ratio_to_minus}, $\pi^{\star}$ is also an optimal solution to \eqref{eq_Simple} and \eqref{eq_DPExpected}.
See Appendix \ref{app_zero_gap} for the details.
\end{proof}
Hence, Theorem~\ref{thm_1} follows from Theorem \ref{thm_zero_gap}.

% !TEX root = ./sampling_BM.tex

\ifreport
\begin{figure}
\centering \includegraphics[width=0.35\textwidth]{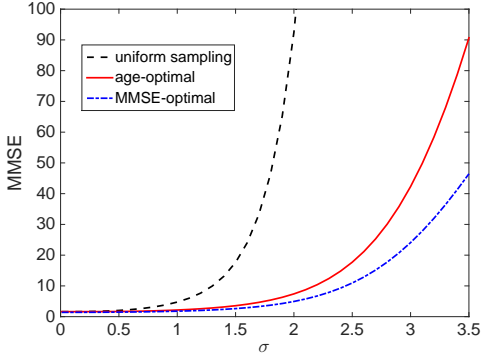} \caption{MMSE vs. the scale parameter  $\sigma$ of \emph{i.i.d.} log-normal channel delay for $f_{\max}=0.8$.}
\label{fig2} \vspace{-0.cm}
\end{figure}
\else
\fi
\section{Numerical Results}
In this section, we evaluate the estimation performance achieved by the following four sampling policies:  

\begin{itemize}
\item[1.] \emph{Uniform sampling:} The policy in \eqref{eq_uniform} with $\beta = f_{\max}$. 
\item[2.] \emph{Zero-wait sampling \cite{2015ISITYates,report_AgeOfInfo2016,SunInfocom2016,KaulYatesGruteser-Infocom2012}:} The sampling policy in \eqref{eq_Zero_wait}, which is feasible when $f_{\max}\geq \EE[Y_i]$.

\item[3.] \emph{Age-optimal sampling \cite{report_AgeOfInfo2016,SunInfocom2016}:} The sampling policy in \eqref{eq_thm2_1} and \eqref{eq_thm2}, which is the optimal solution to  \eqref{eq_age}.

\item[4.] \emph{MMSE-optimal sampling:} The sampling policy in \eqref{eq_opt_solution} and \eqref{eq_thm1}, which is the optimal solution to \eqref{eq_DPExpected}.
\end{itemize}
Let ${\mathsf{mmse}_{\text{uniform}}}$, ${\mathsf{mmse}_{\text{zero-wait}}}$, ${\mathsf{mmse}_{\text{age-opt}}}$, and ${\mathsf{mmse}_{\text{opt}}}$, 
be the MMSEs of uniform sampling, zero-wait sampling, age-optimal sampling,  MMSE-optimal sampling, respectively. 
According to \eqref{eq_compare_aware}, as well as the facts that uniform sampling is  feasible for \eqref{eq_age} and zero-wait sampling is  feasible for \eqref{eq_age} when $f_{\max}\geq \EE[Y_i]$,  
we can obtain  
\begin{align}
&{\mathsf{mmse}_{\text{opt}}}\leq{\mathsf{mmse}_{\text{age-opt}}}\leq{\mathsf{mmse}_{\text{uniform}}}, \nonumber\\
&{\mathsf{mmse}_{\text{opt}}}\leq{\mathsf{mmse}_{\text{age-opt}}}\leq{\mathsf{mmse}_{\text{zero-wait}}},~\text{when $f_{\max}\geq \EE[Y_i]$,} \nonumber
\end{align}
which fit with our numerical results below.

\ifreport
\else
\begin{figure}
\centering \includegraphics[width=0.33\textwidth]{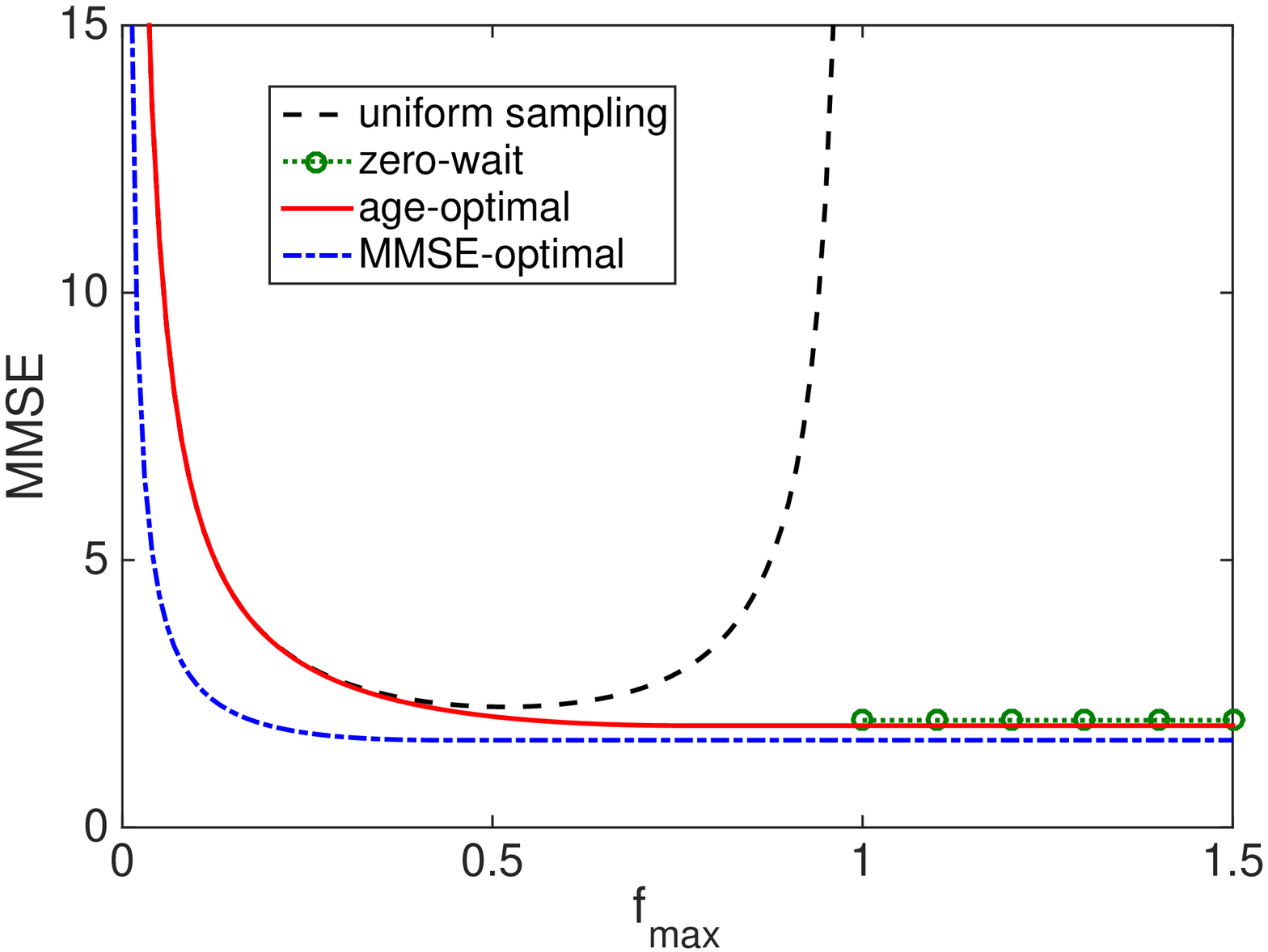} \caption{MMSE vs. ${f_{\max}}$ tradeoff for \emph{i.i.d.} exponential channel delay.}
\label{fig1} \vspace{-1em}
\end{figure}
\fi

Figure \ref{fig1} depicts the tradeoff between MMSE and $f_{\max}$  for \emph{i.i.d.} exponential channel delay with mean $\mathbb{E}[Y_i] = 1/\mu=1$. Hence, the maximum throughput of the channel is $\mu=1$. In this setting, ${\mathsf{mmse}_{\text{uniform}}}$ is characterized by eq. (25) of \cite{KaulYatesGruteser-Infocom2012}, which was obtained using a D/M/1 queueing model. For small values of $f_{\max}$,   age-optimal sampling is similar with uniform sampling, and hence ${\mathsf{mmse}_{\text{age-opt}}}$ and ${\mathsf{mmse}_{\text{uniform}}}$ are of similar values. However, as $f_{\max}$ approaches the maximum throughput $1$, ${\mathsf{mmse}_{\text{uniform}}}$ increases to infinite. 
This is because the queue length in uniform sampling is large at high sampling frequencies, and the samples become stale during their long waiting times in the queue. On the other hand, ${\mathsf{mmse}_{\text{opt}}}$ and ${\mathsf{mmse}_{\text{age-opt}}}$ decrease with respect to $f_{\max}$. The reason is that the set of feasible policies satisfying the constraints in \eqref{eq_DPExpected} and \eqref{eq_age} becomes larger as $f_{\max}$ grows, and hence the optimal  values of \eqref{eq_DPExpected} and \eqref{eq_age} are decreasing in $f_{\max}$. Moreover, the gap between ${\mathsf{mmse}_{\text{opt}}}$ and ${\mathsf{mmse}_{\text{age-opt}}}$ is large for small values of $f_{\max}$. The ratio ${\mathsf{mmse}_{\text{opt}}}/{\mathsf{mmse}_{\text{age-opt}}}$ tends to $1/3$ as $f_{\max}\rightarrow 0$, which is in accordance with \eqref{eq_ratio_MSE}.
As we expected, ${\mathsf{mmse}_{\text{zero-wait}}}$ is larger than ${\mathsf{mmse}_{\text{opt}}}$ and ${\mathsf{mmse}_{\text{age-opt}}}$ when $f_{\max}\geq1$.

%Note that, according to \eqref{eq_thm1} and \eqref{eq_thm2}, the sampling frequencys of the age-optimal and MSE-optimal policies are smaller than $f_{\max}$ when $f_{\max}$ is large. 

\ifreport

\begin{figure}
\centering \includegraphics[width=0.35\textwidth]{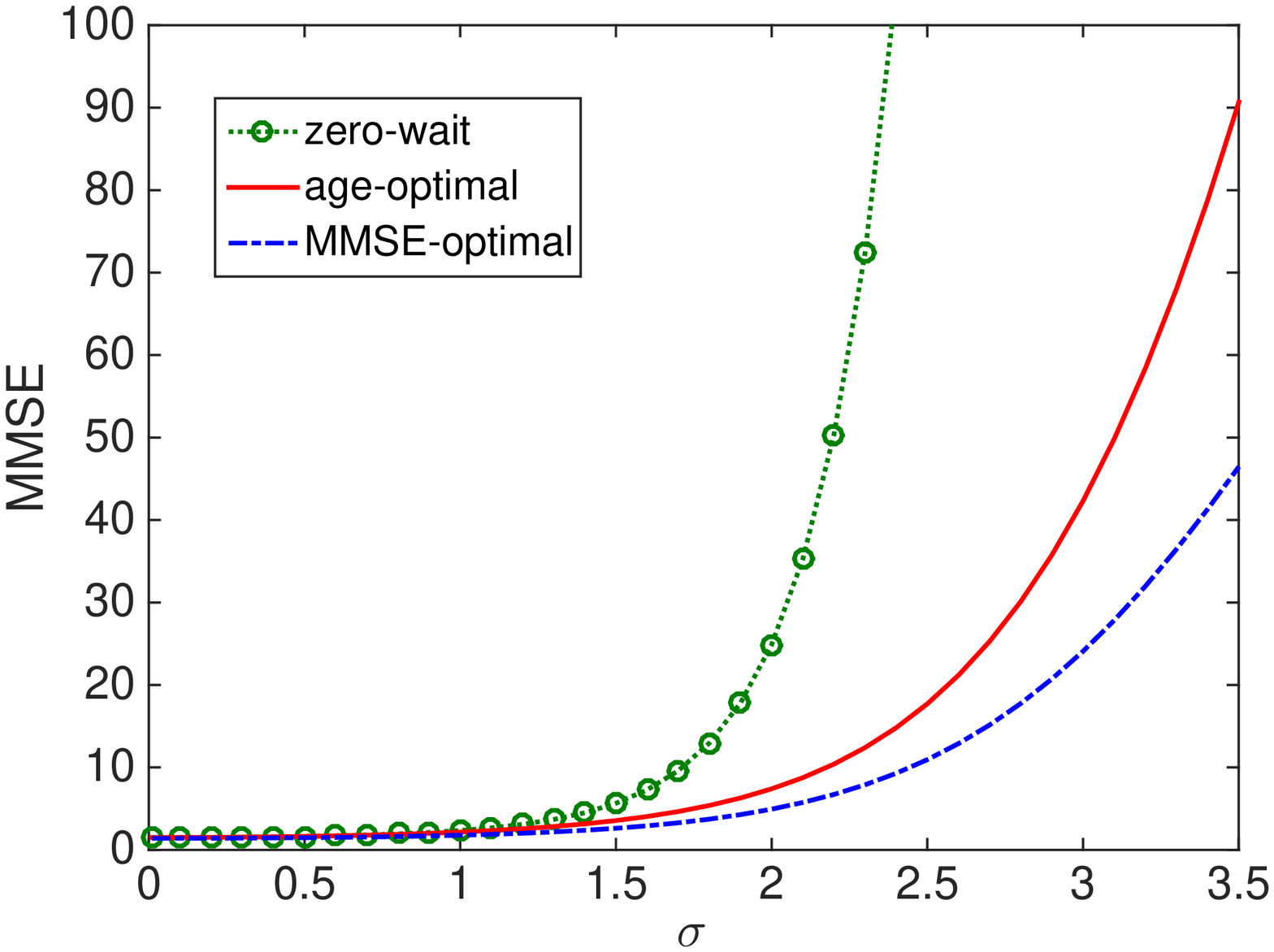} \caption{MMSE vs. the scale parameter  $\sigma$ of \emph{i.i.d.} log-normal channel delay for $f_{\max}=1.5$.}
\label{fig3} \vspace{-0.cm}
\end{figure}

Figure \ref{fig2} and Figure \ref{fig3} illustrate the MMSE of \emph{i.i.d.} log-normal channel delay for $f_{\max}=0.8$ and $f_{\max}=1.5$, respectively, where $Y_i = e^{\sigma X_i}/ \mathbb{E}[e^{\sigma X_i}]$, $\sigma>0$ is the scale parameter of log-normal distribution, and $(X_1,X_2,\ldots)$ are \emph{i.i.d.} Gaussian random variables with zero mean and unit variance. Because $ \mathbb{E}[Y_i] = 1$, the maximum throughput of the channel is  $1$. In Fig. \ref{fig2}, since $f_{\max}<1$, zero-wait sampling is not feasible and hence is not plotted. As the scale parameter $\sigma$ grows, the tail of the log-normal distribution becomes heavier and heavier. We observe that ${\mathsf{mmse}_{\text{uniform}}}$ grows quickly with respect to $\sigma$, much faster than ${\mathsf{mmse}_{\text{opt}}}$ and ${\mathsf{mmse}_{\text{age-opt}}}$. 
In addition, the gap between ${\mathsf{mmse}_{\text{opt}}}$ and ${\mathsf{mmse}_{\text{age-opt}}}$ increases as $\sigma$ grows.
In Fig. \ref{fig3}, because $f_{\max}> 1$, ${\mathsf{mmse}_{\text{uniform}}}$ is infinite and hence is not plotted. We can find that ${\mathsf{mmse}_{\text{zero-wait}}}$ grows quickly with respect to $\sigma$ and is much larger than ${\mathsf{mmse}_{\text{opt}}}$ and ${\mathsf{mmse}_{\text{age-opt}}}$. 
\fi

\section{Conclusion}
In this paper, we have investigated optimal sampling and remote estimation of the Wiener process over a channel with random delay. 
The optimal sampling policy for minimizing the mean square estimation error subject to a sampling frequency constraint has been obtained. We prove that the optimal sampling policy is a threshold policy, and find the optimal threshold. Analytical and numerical comparisons with several important sampling policies, including age-optimal sampling, zero-wait sampling, and classic uniform sampling,  have been provided. 
The results in this paper generalize recent research on ago-of-information by adding a signal model, and can be  also considered a contribution to the rich literature on remote estimation by adding a channel that consists of a queue with random delay. 

\section*{Acknowledgement}
The authors are grateful to Ness B. Shroff and Roy D. Yates for their careful reading of the paper and valuable suggestions.
%we study the optimal sampling strategy for minimizing the mean square estimation error subject to a sampling frequency constraint. We prove that the optimal sampling strategy is a threshold policy, and find the optimal threshold. This threshold is determined by the sampling frequency constraint and how much the signal  varies during the channel delay.

% threshold policies: To avoid wasting time in  queueing, no new sample is taken (i.e., the threshold policy is disabled) until all previous samples are delivered.

%An unexpected consequence is that even in the absence of the sampling frequency constraint, the optimal strategy is not zero-wait sampling, in which a new sample is generated once the previous sample is delivered. Furthermore, if the sampling times are independent of the Wiener signal, the optimal sampling problem reduces to an age-of-information optimization problem solved before. Our comparisons show that the estimation error of the optimal sampling policy is much smaller than those of age-optimal sampling, zero-wait sampling, and the classic uniform sampling.
%\input{./v4/sec_examples}
%\input{./v4/sec_conclusion}
\bibliographystyle{IEEEtran}
\bibliography{ref,AgeofRealtimeMeasurements,sueh}
\appendices

% !TEX root = ./sampling_BM.tex
\section{Proof of \eqref{eq_age_MSE}}\label{app_age}
If $\pi$ is independent of $\{W_t,t\in[0,\infty)\}$, the $S_i$'s and $D_i$'s are independent of $\{W_t,t\in[0,\infty)\}$. Hence,
\begin{align}
&\mathbb{E}\left\{\int_{D_{i}}^{D_{i+1}} (W_t-\hat W_t)^2dt\right\} \nonumber\\
\overset{(a)}{=}& \mathbb{E}\left\{\mathbb{E}\left\{\int_{D_{i}}^{D_{i+1}} (W_t-W_{S_i})^2dt\Bigg|S_i,D_i,D_{i+1}\right\} \right\} \nonumber\\
\overset{(b)}{=}& \mathbb{E}\left\{\int_{D_{i}}^{D_{i+1}} \mathbb{E}\left\{(W_t-W_{S_i})^2|S_i,D_i,D_{i+1}\right\} dt\right\} \nonumber\\
\overset{(c)}{=}& \mathbb{E}\left\{\int_{D_{i}}^{D_{i+1}} \mathbb{E}\left[(W_t-W_{S_i})^2\right] dt\right\} \nonumber\\
\overset{(d)}{=}& \mathbb{E}\left\{\int_{D_{i}}^{D_{i+1}}  (t - S_i) dt\right\} \nonumber\\
\overset{(e)}{=}& \mathbb{E}\left\{\int_{D_{i}}^{D_{i+1}}  \Delta(t) dt\right\}, \nonumber
\end{align} 
where step (a) is due to the law of iterated expectations, step (b) is due to Fubini's theorem, step (c) is because $S_i,D_i,D_{i+1}$ are  independent of the Wiener process, step (d) is due to Wald's identity $\mathbb{E}[W_T^2] = T$ \cite[Theorem 2.48]{BMbook10} and the strong Markov property of the Wiener process \cite[Theorem 2.16]{BMbook10}, and step (e) is due to \eqref{eq_age_def}. By this, \eqref{eq_age_MSE} is proven.

\section{Proofs of \eqref{eq_betaasy} and  \eqref{eq_opt_limit_1}}\label{app_low_delay}
If $f_{\max}\rightarrow 0$,  \eqref{eq_thm1} tells us that
\begin{align}
\mathbb{E}[\max(\beta,W_Y^2)] =\frac{1}{f_{\max}},\nonumber
\end{align}
which implies
\begin{align}
\beta \leq \frac{1}{f_{\max}} \leq \beta+ \mathbb{E}[W_Y^2] = \beta+ \mathbb{E}[Y].\nonumber
\end{align}
Hence,
\begin{align}
\frac{1}{f_{\max}} -\mathbb{E}[Y] \leq \beta \leq \frac{1}{f_{\max}}.\nonumber
\end{align}
If $f_{\max}\rightarrow 0$, \eqref{eq_betaasy} follows. Because $Y$ is independent of the Wiener process, using the law of iterated expectations and the Gaussian distribution of the Wiener process, we can obtain $\mathbb{E}[W_Y^4]=3\mathbb{E}[Y^2]$ and $\mathbb{E}[W_Y^2]=3\mathbb{E}[Y]$. Hence, 
\begin{align}
\beta\leq &~\mathbb{E}[\max(\beta,W_Y^2)] \leq \beta+ \mathbb{E}[W_Y^2]=\beta+ \mathbb{E}[Y],\nonumber\\
\beta^2\leq &~\mathbb{E}[\max(\beta^2,W_Y^4)] \leq \beta^2+ \mathbb{E}[W_Y^4]=\beta^2+ 3\mathbb{E}[Y^2].\nonumber
\end{align}
Therefore, 
\begin{align}\label{eq_eq_betaasy}
\frac{\beta^2}{\beta+ \mathbb{E}[Y]}\leq \frac{\mathbb{E}[\max(\beta^2,W_Y^4)]}{\mathbb{E}[\max(\beta,W_Y^2)]}\leq \frac{\beta^2+ 3\mathbb{E}[Y^2]}{\beta}.
\end{align}
By combining  \eqref{thm_1_obj}, \eqref{eq_betaasy}, and \eqref{eq_eq_betaasy}, \eqref{eq_opt_limit_1} follows in the case of $f_{\max}\rightarrow 0$.

If $d\rightarrow  0$, then $Y \rightarrow 0$ and $W_Y \rightarrow 0$ with probability one. Hence, $\mathbb{E}[\max(\beta,W_Y^2)] \rightarrow \beta$ and $\mathbb{E}[\max(\beta^2,W_Y^4)] \rightarrow \beta^2$. Substituting these
%, together with \eqref{eq_expression0} and \eqref{eq_expression}
 into \eqref{eq_thm1} and \eqref{eq_eq_betaasy}, yields
\begin{align}
\lim_{d\rightarrow  0 }\beta = \frac{1}{f_{\max}},~ \lim_{d\rightarrow  0 }\left\{\frac{\mathbb{E}[\max(\beta^2,W_Y^4)]}{6\mathbb{E}[\max(\beta,W_Y^2)]} + \mathbb{E}[Y]\right\} = \frac{1}{6f_{\max}}.\nonumber
\end{align}
By this, \eqref{eq_betaasy} and \eqref{eq_opt_limit_1} are proven in the case of $d\rightarrow 0$.
This completes the proof.

\section{Proof of \eqref{eq_coro_1}}\label{app_scale}
If $f_{\max}\rightarrow\infty$, the sampling frequency constraint in \eqref{eq_DPExpected} can be removed. By \eqref{eq_thm1}, the optimal $\beta$ is determined by  \eqref{eq_coro_1}.

If $d\rightarrow\infty$, let us consider the equation
\begin{align} \label{eq_a_equation}
\mathbb{E}[\max(\beta,W_Y^2)] \!=\! \frac{\mathbb{E}[\max(\beta^2,W_Y^4)]}{2\beta}.
\end{align}
If $Y$ grows by $a$ times, then $\beta$ and $\mathbb{E}[\max(\beta,W_Y^2)]$ in \eqref{eq_a_equation} both should grow by $a$ times, and $\mathbb{E}[\max(\beta^2,W_Y^4)]$ in \eqref{eq_a_equation} should grow by $a^2$ times. Hence, if $d\rightarrow\infty$, it holds in \eqref{eq_thm1} that 
\begin{align}
\frac{1}{f_{\max}}\leq\frac{\mathbb{E}[\max(\beta^2,W_Y^4)]}{2\beta}
\end{align} and 
the solution to \eqref{eq_thm1} is given by \eqref{eq_coro_1}. This completes the proof.

\section{Proofs of Theorems \ref{lem_zero_wait1} and \ref{lem_zero_wait2}}\label{app_zerowait}
\begin{proof}[Proof of Theorem \ref{lem_zero_wait1}]
The zero-wait policy can be expressed as \eqref{eq_opt_solution} with $\beta =0$. Because $Y$ is independent of the Wiener process, using the law of iterated expectations and the Gaussian distribution of the Wiener process, we can obtain $\mathbb{E}[W_Y^4]=3\mathbb{E}[Y^2]$. 
According to \eqref{eq_coro_1}, $\beta=0$  if and only if $\mathbb{E}[W_Y^4]=3\mathbb{E}[Y^2]=0$ which is equivalent to $Y=0$ with probability one. This completes the proof.
\end{proof}
\begin{proof}[Proof of Theorem \ref{lem_zero_wait2}]
In the one direction, the zero-wait policy can be expressed as \eqref{eq_thm2_1} with $\beta \leq \text{ess}\inf Y$. If the zero-wait policy is optimal, then  the solution to \eqref{eq_coro_2} must satisfy $\beta \leq \text{ess}\inf Y$, which further implies $\beta\leq Y$ with probability one. From this, we can get 
\begin{align}\label{eq_coro_3}
 2  \text{ess}\inf Y \mathbb{E}[Y]\geq 2\beta \mathbb{E}[Y] = {\mathbb{E}[Y^2]},
\end{align}
By this, \eqref{eq_zero_wait2} follows.

In the other direction, if \eqref{eq_zero_wait2} holds, we will show that the zero-wait policy is age-optimal by considering the following two cases. 

\emph{Case 1:} $\mathbb{E}[Y]> 0$. By choosing \begin{align}\label{eq_coro_5}
\beta = \frac{\mathbb{E}[Y^2]}{2\mathbb{E}[Y]},
\end{align}we can get $\beta \leq \text{ess}\inf Y$
and hence
\begin{align}\label{eq_coro_6}
\beta \leq  Y
\end{align}
with probability one.
According to \eqref{eq_coro_5} and \eqref{eq_coro_6}, such a $\beta$ is the solution to \eqref{eq_coro_2}. Hence, the zero-wait policy expressed by \eqref{eq_thm2_1} with $\beta \leq \text{ess}\inf Y$ is the age-optimal policy. 

\emph{Case 2:} $\mathbb{E}[Y]= 0$ and hence $Y=0$ with probability one. In this case, $\beta =0$ is the solution to \eqref{eq_coro_2}. Hence, the zero-wait policy expressed by \eqref{eq_thm2_1} with $\beta =0$ is the age-optimal policy. 

Combining these two cases, the proof is completed.
\end{proof}

\section{Proof of Lemma \ref{lem_zeroqueue}}\label{app_zeroqueue}
%We prove Lemma \ref{lem_zeroqueue} in two steps:

%This happens if % happen in two scenarios: (i) sample $i$ is generated when the channel is busy sending another sample 

Suppose that in policy $\pi$, sample $i$ is generated when the channel is busy sending another sample, and hence sample $i$  
needs to wait for some time before submitted to the channel, i.e., $S_i<G_i$. 
%and (ii) following the given queue service discipline, sample $i$ is stored in the queue after it is generated, and another sample waiting in the queue is submitted to the channel. 
Let us consider a \emph{virtual} sampling policy $\pi' = \{S_0,\ldots, S_{i-1},G_i, S_{i+1},\ldots\}$. We call policy $\pi'$ a virtual policy because the generation time of sample $i$ in policy $\pi'$ is $S'_i =G_i$ and it may happen that  $S'_i > S_{i+1}$. However, this will not affect our  proof below. We will show that the MMSE of  policy $\pi'$ is smaller than that of policy $\pi =\{S_0,\ldots, S_{i-1},S_i, S_{i+1},\ldots\}$. 
%The Wiener process $\{W_t: t\in[0, \infty)\}$ remain the same under any policy. In addition, because 
%For any given non-preemptive queueing discipline, 

Note that the Wiener process $\{W_t: t\in[0, \infty)\}$ does not change according to the sampling policy, and the sample delivery times $\{D_1, D_2,\ldots\}$ remain the same in policy $\pi$ and policy $\pi'$. %Because the sample delivery times $\{D_1, D_2,\ldots\}$ are uniquely determined by the channel states, $\{D_1, D_2,\ldots\}$ remains the same in both policy $\pi$ and policy $\pi'$.
Hence, the only difference between policies $\pi$ and $\pi'$ is that \emph{the generation time of sample $i$ is postponed from $S_i$ to $G_i$}. 
The MMSE estimator of policy $\pi$ is given by \eqref{eq_esti} and the  MMSE estimator of policy $\pi'$ is given by
\begin{align}\label{eq_esti_pi1}
\hat W_t = \left\{\begin{array}{l l} 0,& t\in[0,D_1);\\
W_{G_i},& t\in[D_i,D_{i+1}); \\
W_{S_j},& t\in[D_j,D_{j+1}),~j\neq i, j\geq 1. \\
\end{array}\right.
\end{align}
%By comparison, the difference between policies $\pi$ and $\pi'$ only lies in the estimated signal during $[D_i,D_{i+1})$.
Because $S_i\leq G_i\leq D_{i} \leq D_{i+1}$, by the strong Markov property of the Wiener process \cite[Theorem 2.16]{BMbook10}, $\int_{D_i} ^{D_{i+1}}2 [W_t - W_{G_i}] dt$ and $W_{G_i} - W_{S_i}$ are mutually independent. Hence,
\begin{align}
&\mathbb{E}\left\{\int_{D_i} ^{D_{i+1}} (W_t-W_{S_i})^2dt\right\} \nonumber\\
= & \mathbb{E}\left\{\int_{D_i} ^{D_{i+1}} (W_t-W_{G_i})^2 +(W_{G_i}- W_{S_i})^2 dt\right\} \nonumber\\
& + \mathbb{E}\left\{\int_{D_i} ^{D_{i+1}}2 (W_t-W_{G_i}) (W_{G_i}- W_{S_i}) dt\right\} \nonumber\\ 
= & \mathbb{E}\left\{\int_{D_i} ^{D_{i+1}} (W_t-W_{G_i})^2 +(W_{G_i}- W_{S_i})^2 dt\right\} \nonumber\\
& + \mathbb{E}\left\{\int_{D_i} ^{D_{i+1}}2 (W_t- W_{G_i}) dt\right\}\mathbb{E}[W_{G_i}-W_{S_i}] . \nonumber
\end{align}
Note that the channel is busy whenever there exist some generated samples that are not delivered to the estimator.
Hence, during the time interval $[S_i, G_i]$, the channel is busy sending some samples generated before $S_i$ in policy $\pi$. 
Because 
$\mathbb{E}[Y_j^2]<\infty$, we can get $\mathbb{E}[Y_j]<\infty$ and
\begin{align}
\mathbb{E}[G_i - S_i]\leq \mathbb{E}\bigg[\sum_{j=1}^{i-1} Y_j\bigg]<\infty.\nonumber
\end{align}
By Wald's identity \cite[Theorem 2.44]{BMbook10}, we have
$\mathbb{E}[W_{G_i}-W_{S_i}]=0$ and hence
\begin{align}
&\mathbb{E}\left\{\int_{D_i} ^{D_{i+1}} (W_t-W_{S_i})^2dt\right\}\nonumber\\
\geq &\mathbb{E}\left\{\int_{D_i} ^{D_{i+1}} (W_t-W_{G_i})^2  dt\right\}.
\end{align}
Therefore, the MMSE of  policy $\pi'$ is smaller than that of policy $\pi$. 

By repeating the above arguments for all samples $i$ satisfying $S_i<G_i$, one can show that policy $\pi'' = \{S_0,G_1,\ldots, G_{i-1},G_i, G_{i+1},\ldots\}$ is better than policy $\pi =\{S_0,S_1,\ldots, S_{i-1},S_i, S_{i+1},\ldots\}$.
%it is better to wait until the channel becomes idle, and then take each new sample and send it out immediately. 
This completes the proof.

\section{Proof of Lemma  \ref{lem_ratio_to_minus}}\label{app_ratio_to_minus}
Part (a) is proven in  two steps:

\emph{Step 1:} We will prove that {$\mathsf{mmse}_{\text{opt}} \leq c $ if and only if $p(c)\leq 0$}.

If $\mathsf{mmse}_{\text{opt}} \leq c $, then there exists a policy $\pi= (Z_0,Z_1,\ldots)\in\Pi_1$ that is feasible for both \eqref{eq_Simple}
and \eqref{eq_SD}, which satisfies  
\begin{align}\label{eq_step1}
\lim_{n\rightarrow \infty}\frac{\sum_{i=0}^{n-1}\mathbb{E}\left[\int_{D_{i}}^{D_{i+1}} (W_t-W_{S_{i}})^2dt\right]}{ \sum_{i=0}^{n-1} \mathbb{E}\left[Y_i\!+\!Z_i\right]}\leq c.
\end{align}
Hence,
\begin{align}\label{eq_step2}
\lim_{n\rightarrow \infty}\frac{\frac{1}{n}\sum_{i=0}^{n-1}\mathbb{E}\left[\int_{D_{i}}^{D_{i+1}} (W_t-W_{S_{i}})^2dt - c(Y_i+Z_i)\right]}{\frac{1}{n} \sum_{i=0}^{n-1} \mathbb{E}\left[Y_i\!+\!Z_i\right]}\leq0.
\end{align}
Because the inter-sampling times $T_i=Y_i+Z_i$ are regenerative, the renewal theory \cite{Ross1996} tells us that the limit 
$\lim_{n\rightarrow \infty}\frac{1}{n} \sum_{i=0}^{n-1} \mathbb{E}\left[Y_i\!+\!Z_i\right]$ exists and is positive.  By this, we get 
\begin{align}\label{eq_step3}
\lim_{n\rightarrow \infty}\frac{1}{n}\sum_{i=0}^{n-1}\mathbb{E}\left[\int_{D_{i}}^{D_{i+1}} (W_t-W_{S_{i}})^2dt - c(Y_i+Z_i)\right] \leq 0.
\end{align}
Therefore, $p(c)\leq 0$.

On the reverse direction, if $p(c)\leq 0$, then there exists a policy $\pi= (Z_0,Z_1,\ldots)\in\Pi_1$ that is feasible for both \eqref{eq_Simple}
and \eqref{eq_SD}, which satisfies \eqref{eq_step3}. From \eqref{eq_step3}, we can derive \eqref{eq_step2} and \eqref{eq_step1}. Hence, $\mathsf{mmse}_{\text{opt}} \leq c $. By this, we have proven that {$\mathsf{mmse}_{\text{opt}} \leq c $ if and only if $p(c)\leq 0$}. 

\emph{Step 2:} We needs to prove that $\mathsf{mmse}_{\text{opt}} < c $ if and only if $p(c)< 0$. This statement can be proven by using the arguments in \emph{Step 1}, in which ``$\leq$'' should be replaced by ``$<$''. Finally, from the statement of \emph{Step 1}, it immediately follows that $\mathsf{mmse}_{\text{opt}} > c $ if and only if $p(c)> 0$. This completes the proof of part (a). 

Part (b): We first show that each optimal solution to \eqref{eq_Simple} is  an optimal solution to  \eqref{eq_SD}. 
By the claim of part (a), $p(c)= 0$ is equivalent to $\mathsf{mmse}_{\text{opt}} = c $. Suppose that policy $\pi= (Z_0,Z_1,\ldots)\in\Pi_1$ is an optimal solution to  \eqref{eq_Simple}.
Then, $\mathsf{mmse}_{\pi} = \mathsf{mmse}_{\text{opt}} = c $. Applying this in the arguments of \eqref{eq_step1}-\eqref{eq_step3}, we can show that policy $\pi$ satisfies
\begin{align}
\lim_{n\rightarrow \infty}\frac{1}{n}\sum_{i=0}^{n-1}\mathbb{E}\left[\int_{D_{i}}^{D_{i+1}} (W_t-W_{S_{i}})^2dt - c(Y_i+Z_i)\right] = 0.\nonumber
\end{align}   
This and $p(c)= 0$ imply that  policy $\pi$ is an optimal solution to  \eqref{eq_SD}. 

Similarly, we can prove that each optimal solution to \eqref{eq_SD} is  an optimal solution to  \eqref{eq_Simple}. By this, part (b) is proven.

\section{Proof of Lemma \ref{lem_decompose}}\label{app_decompose}

Because the $Y_i$'s are \emph{i.i.d.}, $Z_i$ is independent of $Y_{i+1}, Y_{i+2},\ldots$, and the strong Markov property of the Wiener process \cite[Theorem 2.16]{BMbook10}, in the Lagrangian $L(\pi;\lambda)$ the term related to $Z_i$ is
\begin{align}\label{eq_decomposed_term}
&\mathbb{E}\left[\int_{S_{i}+Y_i}^{S_{i}+Y_i+Z_i +Y_{i+1}} (W_t-W_{S_{i}})^2dt\!-\! (c+\lambda)(Y_i+Z_i)\right],
%\nonumber\\
%=&\mathbb{E}\Bigg[\int_{S_{i}+Y_i}^{S_{i}+Y_i+Z_i +Y_{i+1}} [(W_{S_{i}+Y_i}\!-\!W_{S_{i}}) \!+\! (W_{t}\!-\!W_{S_{i}+Y_i})]^2dt\nonumber\\
%&~~~~~~~~~~~~~~~~~\!-(c+\lambda)(Y_i+Z_i)\Bigg].
\end{align}
which is determined by 
the control decision $Z_i$ and the recent information of the system $\mathcal{I}_i = (Y_i, (W_{S_{i}+t}-W_{S_{i}},$ $t \geq 0))$. According to \cite[p. 252]{Bertsekas2005bookDPVol1} and \cite[Chapter 6]{Kumar1986}, $\mathcal{I}_i$ is the \emph{sufficient statistic} for determining  $Z_i$ in \eqref{eq_primal}. Therefore, there exists an optimal policy $(Z_0,Z_1,\ldots)$ in which $Z_i$ is determined based on only $\mathcal{I}_i$, which is independent of $(W_t: t\in[0,S_i])$. 
This completes the proof.

%In addition, because the $Y_i$'s are \emph{i.i.d.} and the strong Markov property of the Wiener process, the $Z_i$'s in this optimal policy $\pi_{\text{opt}}$ are  \emph{i.i.d.} Similarly, the $(W_{S_{i}+Y_i+Z_i}-W_{S_i})$'s in this optimal policy $\pi_{\text{opt}}$ are  \emph{i.i.d.}

\section{Proof of Lemma \ref{lem_stop}}\label{applem_stop}
%We first prove Lemma \ref{lem_stop}:
%\begin{proof}[Proof of Lemma \ref{lem_stop}]
According to Theorem 2.51 and Exercise 2.15  of \cite{BMbook10}, 
%$X_1 (t) =\int_0^t W_s^2 ds - 1/6 W_t^4$ and  $X_2(t)=\int_0^t W_s^2 ds-t W_t^2 + t^2/2$.
$W_t^4 - 6\int_0^t W_s^2 ds$ and $W_t^4 - 6t W_t^2 + 3t^2$ are two martingales of the Wiener process $\{W_t,t\in[0,\infty)\}$. Hence, $\int_0^t W_s^2 ds-t W_t^2 + t^2/2$ is also a martingale of the Wiener process. 

Because the minimum of two stopping times is a stopping time and constant times are stopping times \cite{Durrettbook10}, it follows that $t\wedge \tau$ is a bounded stopping time for every $t\in[0,\infty)$, where $x\wedge y = \min[x,y]$. Then, it follows from Theorem 8.5.1 of \cite{Durrettbook10} that for every $t\in[0,\infty)$
\begin{align}
\!\!\!\!\mathbb{E}\!\left[\int_0^{t \wedge \tau }\!\!\! W_s^2 ds\right] &\!=\!  \frac{1}{6}\mathbb{E}\left[ W_{t \wedge \tau}^4 \right] \label{eq_martigale1}\\
&= \mathbb{E}\left[  (t \wedge \tau) W_{t \wedge \tau}^2 \!- \frac{1}{2} (t \wedge \tau)^2\right]\!.\!\!\! \label{eq_martigale2}
\end{align}
Notice that $\int_0^{t \wedge \tau} W_s^2 ds$ is positive and increasing with respect to $t$. By applying the monotone convergence theorem \cite[Theorem 1.5.5]{Durrettbook10}, we can obtain 
\begin{align}
\lim_{t\rightarrow \infty }\mathbb{E}\left[\int_0^{t \wedge \tau} W_s^2 ds\right] = \mathbb{E}\left[\int_0^{\tau} W_s^2 ds\right].\nonumber
\end{align}

The remaining task is to show that 
\begin{align}
\lim_{t\rightarrow \infty } \mathbb{E}\left[ W_{t \wedge \tau}^4 \right] = \mathbb{E}\left[ W_{\tau}^4 \right]. \label{eq_goal}
\end{align}
Towards this goal, we combine \eqref{eq_martigale1} and \eqref{eq_martigale2}, and apply  Cauchy-Schwarz inequality to get
\begin{align}
& \mathbb{E}\left[ W_{t \wedge \tau}^4 \right]\nonumber \\
= & \mathbb{E}\left[ 6 (t \wedge \tau) W_{t \wedge \tau}^2 - 3 (t \wedge \tau)^2\right] \nonumber \\
\leq & 6 \sqrt{\mathbb{E}\left[ ( t \wedge \tau)^2\right] \mathbb{E}\left[ W_{t \wedge \tau}^4\right]}  - 3\mathbb{E}\left[ (t \wedge \tau)^2\right].\nonumber
\end{align} 
Let $x = \sqrt{\mathbb{E}\left[ W_{t \wedge \tau}^4 \right]/\mathbb{E}\left[ (t \wedge \tau)^2\right]}$, then $x^2 -6 x +3 \leq 0$.
By the roots and properties of quadratic functions, we obtain $3 - \sqrt{6}\leq x\leq 3 + \sqrt{6}$ and hence
\begin{align}
&  \mathbb{E}\left[ W_{t \wedge \tau}^4 \right]    \leq (3 + \sqrt{6})^2 \mathbb{E}  \left[ (t \wedge \tau)^2\right] \leq (3 + \sqrt{6})^2 \mathbb{E}\left[  \tau^2\right]< \infty. \nonumber 
\end{align} 
Then, we use Fatou's lemma \cite[Theorem 1.5.4]{Durrettbook10} to derive
\begin{align}\label{eq_fatou}
& \mathbb{E}\left[ W_{\tau}^4 \right] \nonumber \\
 = & 
\mathbb{E}\left[ \lim_{t\rightarrow \infty } W_{t \wedge \tau}^4 \right]\nonumber \\
\leq & \liminf_{t\rightarrow \infty }  \mathbb{E}\left[ W_{t \wedge \tau}^4 \right] \nonumber \\
\leq & (3 + \sqrt{6})^2 \mathbb{E}\left[  \tau^2\right]< \infty. 
\end{align} 
Further, by \eqref{eq_fatou} and Doob's maximal inequality \cite[Theorem 12.30]{BMbook10} and \cite[Theorem 5.4.3]{Durrettbook10}, 
\begin{align}
& \mathbb{E}\left[ \sup_{t\in[0,\infty)} W_{t\wedge \tau}^4 \right]  \nonumber \\
=& \mathbb{E}\left[ \sup_{t\in[0,\tau]} W_t^4 \right]  \nonumber \\
\leq & \left(\frac{4}{3}\right)^4 \mathbb{E}\left[ W_{\tau}^4  \right]  < \infty. 
\end{align} 
Because $W_{t\wedge \tau}^4 \leq \sup_{t\in[0,\infty)} W_{t\wedge \tau}^4$ and $\sup_{t\in[0,\infty)} W_{t\wedge \tau}^4$ is integrable, \eqref{eq_goal} follows from dominated convergence theorem \cite[Theorem 1.5.6]{Durrettbook10}. This completes the proof.
%\end{proof}

\section{Proof of  \eqref{eq_integral}}\label{app_integral}

By using \eqref{eq_policyspace} and the condition that $Z_i$ is independent of $(W_t, t\in[0,W_{S_i}])$, we obtain that for
given $Y_i$ and $Y_{i+1}$, $Y_i$ and $Y_i+Z_i+Y_{i+1}$ are stopping times of the time-shifted Wiener process $\{W_{S_{i}+t}-W_{S_{i}},t\geq0\}$. Hence,

%By using \eqref{eq_esti},  \eqref{eq_stopping}, and Lemma \ref{lem_stop}, 
%We can obtain
\begin{align}\label{eq_Q_i}
{=}& \mathbb{E}\left\{\int_{D_{i}}^{D_{i+1}} (W_t-W_{S_i})^2dt\right\} \nonumber\\
\overset{}{=} & \mathbb{E}\left\{\int_{Y_i}^{Y_i+Z_i+Y_{i+1}} (W_{S_{i}+t}-W_{S_{i}})^2dt\right\} \nonumber\\
\overset{(a)}{=} & \mathbb{E}\left\{\mathbb{E}\left\{\int_{Y_i}^{Y_i+Z_i+Y_{i+1}} (W_{S_{i}+t}-W_{S_{i}})^2dt\Bigg| Y_i, Y_{i+1}\right\}\right\} \nonumber\\
\overset{(b)}{=} &  \frac{1}{6}\mathbb{E}\left\{\mathbb{E}\left\{(W_{S_{i}+Y_i+Z_i+Y_{i+1}}-W_{S_{i}})^4\Bigg| Y_i, Y_{i+1}\right\}\!\!\right\} \nonumber\\
&-\frac{1}{6} \mathbb{E}\left\{\mathbb{E}\left\{(W_{S_{i}+Y_i}-W_{S_{i}})^4\Bigg| Y_i, Y_{i+1}\right\}\!\!\right\} \nonumber\\
\overset{(c)}{=} & \frac{1}{6}\mathbb{E}\left[ (W_{S_{i}+Y_i+Z_i+Y_{i+1}}\!-\!W_{S_{i}})^4 \right] - \frac{1}{6}\mathbb{E}\left[ (W_{S_{i}+Y_i}\!-\!W_{S_{i}})^4 \right]\!\!, 
\end{align}
where   step (a) and step (c) are due to the law of iterated expectations, and step (b) is due to Lemma \ref{lem_stop}. %\eqref{eq_stopping}, together with the condition that , suggests that, . 
%In addition,%Hence, we can use Lemma \ref{lem_stop}.
Because $S_{i+1} = {S_{i}+Y_i+Z_i}$, we have
\begin{align}
&\mathbb{E}\left[ (W_{S_{i}+Y_i+Z_i+Y_{i+1}}-W_{S_{i}})^4 \right] \nonumber\\
=&\mathbb{E}\left\{[ (W_{S_{i}+Y_i+Z_i}-W_{S_i}) +(W_{S_{i+1}+Y_{i+1}}-W_{S_{i+1}})]^4 \right\} \nonumber\\
%=&\mathbb{E}\left\{\mathbb{E}\left\{[ B(Y_i+Z_i) + (B(Y_i+Z_i+Y_{i+1}) - B(Y_i+Z_i))]^4 \right.\right.~\nonumber\\
%&\left.\left. | Y_i, Y_{i+1}\right\} \right\} \nonumber\\
%\overset{(a)}{=}&\mathbb{E}\left\{\mathbb{E}\left\{[ B(Y_i+Z_i) + B'(Y_{i+1})]^4 | Y_i, Y_{i+1}\right\}\right\} \nonumber\\
%=&\mathbb{E}\left\{[ B(Y_i+Z_i) + B(Y_{i+1})]^4 \right\} \nonumber\\
=&\mathbb{E}\left[(W_{S_{i}+Y_i+Z_i}-W_{S_i})^4\right] \nonumber\\
&+ 4\mathbb{E}\left[(W_{S_{i}+Y_i+Z_i}-W_{S_i})^3(W_{S_{i+1}+Y_{i+1}}-W_{S_{i+1}})\right]  \nonumber\\
&+6\mathbb{E}\left[ (W_{S_{i}+Y_i+Z_i}-W_{S_i})^2(W_{S_{i+1}+Y_{i+1}}-W_{S_{i+1}})^2\right]   \nonumber\\
&+ 4\mathbb{E}\left[(W_{S_{i}+Y_i+Z_i}-W_{S_i})(W_{S_{i+1}+Y_{i+1}}-W_{S_{i+1}})^3\right] \nonumber\\
&+ \mathbb{E}\left[(W_{S_{i+1}+Y_{i+1}}-W_{S_{i+1}})^4\right] \nonumber\\
=&\mathbb{E}\left[(W_{S_{i}+Y_i+Z_i}-W_{S_i})^4\right] \nonumber\\
&+ 4\mathbb{E}\left[(W_{S_{i}+Y_i+Z_i}-W_{S_i})^3\right]\mathbb{E}\left[(W_{S_{i+1}+Y_{i+1}}-W_{S_{i+1}})\right]  \nonumber\\
&+6\mathbb{E}\left[ (W_{S_{i}+Y_i+Z_i}-W_{S_i})^2\right] \mathbb{E}\left[(W_{S_{i+1}+Y_{i+1}}-W_{S_{i+1}})^2\right]   \nonumber\\
&+ 4\mathbb{E}\left[(W_{S_{i}+Y_i+Z_i}-W_{S_i})\right]\mathbb{E}\left[(W_{S_{i+1}+Y_{i+1}}-W_{S_{i+1}})^3\right] \nonumber\\
&+ \mathbb{E}\left[(W_{S_{i+1}+Y_{i+1}}-W_{S_{i+1}})^4\right], \nonumber
\end{align}
where in the last equation we have used the fact that $Y_{i+1}$ is independent of $Y_i$ and $Z_i$, and the strong Markov property of the Wiener process \cite[Theorem 2.16]{BMbook10}.  
Because
\begin{align}
&\mathbb{E}\left[(W_{S_{i+1}+Y_{i+1}}-W_{S_{i+1}})^3 | Y_{i+1}\right]\nonumber\\
=&\mathbb{E}\left[(W_{S_{i+1}+Y_{i+1}}-W_{S_{i+1}}) | Y_{i+1}\right] =  0,\nonumber
\end{align}
by the law of iterated expectations, we have $$\mathbb{E}\left[(W_{S_{i+1}+Y_{i+1}}\!\!-\!W_{S_{i+1}})^3\right] \!=\! \mathbb{E}\left[(W_{S_{i+1}+Y_{i+1}}\!\!-\!W_{S_{i+1}})\right] =0.$$ In addition,  Wald's identity tells us that $\mathbb{E}\left[ W_{\tau}^2\right] =  \mathbb{E}\left[\tau\right] $ for any stopping time $\tau$ with $\mathbb{E}\left[\tau\right] <\infty$. Hence,  
\begin{align}
&\mathbb{E}\left[ (W_{S_{i}+Y_i+Z_i+Y_{i+1}}-W_{S_{i}})^4 \right] \nonumber\\
=&\mathbb{E}\left[(W_{S_{i}+Y_i+Z_i}-W_{S_i})^4\right]\! +\!6\mathbb{E}\left[Y_i\!+\!Z_i\right] \mathbb{E}\left[Y_{i+1}\right]\!\nonumber\\
&+ \!\mathbb{E}\left[(W_{S_{i+1}+Y_{i+1}}-W_{S_{i+1}})^4\right]\!.  
\end{align}
Finally, because $(W_{S_{i}+t}-W_{S_{i}})$ and $(W_{S_{i+1}+t}-W_{S_{i+1}})$ are both Wiener processes, and the $Y_i$'s are \emph{i.i.d.}, 
\begin{align}\label{eq_Q_i1}
\mathbb{E}\left[(W_{S_{i}+Y_{i}}-W_{S_{i}})^4\right] = \mathbb{E}\left[(W_{S_{i+1}+Y_{i+1}}\!-\!W_{S_{i+1}})^4\right].
\end{align}
Combining  \eqref{eq_Q_i}-\eqref{eq_Q_i1}, yields \eqref{eq_integral}.

\section{Proof of Theorem  \ref{thm_solution_form}}\label{app_solution_form}
%Theorem  \ref{thm_solution_form} is proven in three steps:

%\emph{Step 1: We first prove \eqref{eq_opt_stopping}.} 
%If {$c = c_{\text{opt}}$, 
By \eqref{eq_integral}, \eqref{eq_decomposed_term}  can be rewritten as 

%Lagrangian function $L(\pi;\lambda)$ can be decomposed into a sequence of terms, and each term can be expressed as 
\begin{align}\label{eq_decomposed_term1}
&\mathbb{E}\left[\int_{S_{i}+Y_i}^{S_{i}+Y_i+Z_i +Y_{i+1}} (W_t-W_{S_{i}})^2dt\!-\! (c+\lambda)(Y_i+Z_i)\right]\nonumber\\
=&\frac{1}{6}(W_{S_{i}+Y_i+Z_i}-W_{S_i})^4\!-\! \frac{\beta}{3}(Y_i\!+\!Z_i) \nonumber\\
=&\frac{1}{6}[(W_{S_{i}+Y_i}\!-\!W_{S_{i}}) \!+\! (W_{S_{i}+Y_i+Z_i}\!-\!W_{S_{i}+Y_i})]^4\!-\! \frac{\beta}{3}(Y_i \!+\! Z_i).\nonumber\\
\end{align}
Because the $Y_i$'s are \emph{i.i.d.}  and the strong Markov property of the Wiener process \cite[Theorem 2.16]{BMbook10}, the term in \eqref{eq_decomposed_term1} is determined by the control decision $Z_i$ and the information $\mathcal{I}_i' = (W_{S_{i}+Y_i}-W_{S_{i}}, Y_i, (W_{S_{i}+Y_i+t}-W_{S_{i}+Y_i},$ $t \geq 0))$. According to \cite[p. 252]{Bertsekas2005bookDPVol1} and \cite[Chapter 6]{Kumar1986}, $\mathcal{I}_i'$ is the \emph{sufficient statistic} for determining the waiting time $Z_i$ in \eqref{eq_primal}. Therefore, there exists an optimal policy $(Z_0,Z_1,\ldots)$ in which $Z_i$ is determined based on only $\mathcal{I}_i'$. %, but not the history transmission durations $(Y,\ldots, Y_{i-1})$ and signal values $(W_t: t\in[0,S_i+Y_i])$. 
By this, \eqref{eq_primal} is decomposed into a sequence of per-sample control problems \eqref{eq_opt_stopping}.
In addition, because the $Y_i$'s are \emph{i.i.d.} and the strong Markov property of the Wiener process, the $Z_i$'s in this optimal policy are  \emph{i.i.d.} Similarly, the $(W_{S_{i}+Y_i+Z_i}-W_{S_i})$'s in this optimal policy are  \emph{i.i.d.}

\section{Proof of Lemma  \ref{lem1_stop}}\label{app_optimal_stopping}
Case 1: If $b^2 \geq \beta$, then \eqref{eq_opt_stop_solution} tells us that 
\begin{align}\label{eq_zero}
\tau^* = 0
\end{align}
 and 
\begin{align}\label{eq_case1}
u(x) = \mathbb{E}[g(X_0)|X_0=x] = g(x) = \beta s  - \frac{1}{2}b^4.
\end{align}
Case 2: If $b^2 < \beta$, then $\tau^* > 0$ and $(b+ W_{\tau^*})^2 = \beta$. Invoking Theorem 8.5.5 in \cite{Durrettbook10}, yields 
\begin{align}\label{eq_expectation}
\mathbb{E}_x \tau^* = - (\sqrt{\beta} - b)(-\sqrt{\beta} -b) = \beta - b^2. 
\end{align}
Using this, we can obtain
\begin{align}\label{eq_case2}
u(x) &= \mathbb{E}_x g(X(\tau^*)) \nonumber\\
    &= \beta (s + \mathbb{E}_x \tau^*) - \frac{1}{2} \mathbb{E}_x \left[(b+ W_{\tau^*})^4\right] \nonumber\\
& = \beta (s + \beta - b^2)  - \frac{1}{2} \beta^2\nonumber\\
& = \beta s + \frac{1}{2}\beta^2 - b^2\beta.
\end{align}
Hence, in Case 2,
\begin{align}
u(x)-g(x)=\frac{1}{2}\beta^2 - b^2\beta + \frac{1}{2}b^4 =\frac{1}{2} (b^2-\beta)^2\geq 0.\nonumber
\end{align}
By combining these two cases, Lemma \ref{lem1_stop} is proven. %The expression of $u(x)$ can be obtained from \eqref{eq_case1} and \eqref{eq_case2}. This completes the proof.

\section{Proof of Lemma  \ref{lem2_stop}}\label{app_optimal_stopping1}
%This proof is motivated by the proof of \cite[Theorem 10.3]{Peskir2006}.

The function $u(s,b)$ is continuous differentiable in $(s,b)$. In addition, $\frac{\partial^2}{\partial^2b}u(s,b)$ is continuous everywhere but at $b = \pm\sqrt{\beta}$. %Since the Lebesgue measure of those time $t$ for which $W_t = \pm\sqrt{\beta}$ is zero, 
%Moreover, the function $u(s,b)$ is concave. 
By the It\^{o}-Tanaka-Meyer formula \cite[Theorem 7.14 and Corollary 7.35]{BMbook10}, we obtain that almost surely 
%\cite[eq. (3.3.23)]{Peskir2006}
\begin{align}\label{eq_ito}
&u(s+t, b+W_t) - u(s, b) \nonumber\\
=& \int_0^t \frac{\partial}{\partial b} u(s+r,b+W_{r}) dW_{r}  \nonumber\\
  &+ \int_0^t \frac{\partial}{\partial s} u(s+r,b+W_{r}) d r \nonumber\\
  &+ \frac{1}{2}\int_{-\infty}^\infty L^{a}(t)\frac{\partial^2}{\partial b^2} u(s+r,b+a)da , 
\end{align}
where $L^a(t)$ is the local time that the Wiener process spends at the level $a$, i.e.,
\begin{align}
L^a(t) = \lim_{\epsilon \downarrow 0} \frac{1}{2\epsilon}\int_0^t 1_{\{| W_s -a|\leq \epsilon\}} ds,
\end{align}
and $1_A$ is the indicator function of event $A$.
By the property of local times of the Wiener process \cite[Theorem 6.18]{BMbook10}, we obtain that almost surely
\begin{align}\label{eq_ito}
&u(s+t, b+W_t) - u(s, b) \nonumber\\
=& \int_0^t \frac{\partial}{\partial b} u(s+r,b+W_{r}) dW_{r}  \nonumber\\
  &+ \int_0^t \frac{\partial}{\partial s} u(s+r,b+W_{r}) d r \nonumber\\
  &+ \frac{1}{2}\int_0^t \frac{\partial^2 }{\partial b^2} u(s+r,b+W_{r}) d r. 
 \end{align}
 
Because
 \begin{align}
 \frac{\partial}{\partial b} u(s,b)  = \left\{\begin{array}{l l} 
-2b^3, &\text{if}~ b^2 \geq \beta;
\vspace{0.5em}\\
- 2\beta b, &\text{if}~ b^2 < \beta,
\end{array}\right.\nonumber
\end{align}
we can obtain that for all $t\geq 0$ and all $x = (s,b)\in \mathbb{R}^2$ 
\begin{align}\label{eq_bound}
&\mathbb{E}_x\left\{\int_0^t \left[\frac{\partial}{\partial b} u(s+r,b+W_{r})\right]^2 dr \right\} <\infty.
\end{align}
This and Thoerem 7.11 of \cite{BMbook10} imply that $\int_0^t \frac{\partial}{\partial b} u(s+r,b+W_{r}) dW_{r}$ is a martingale and 
\begin{align}\label{eq_ito_integral}
\mathbb{E}_x\left[\int_0^t \frac{\partial}{\partial b} u(s+r,b+W_{r}) dW_{r}\right] =0,~ \forall~ t\geq0.
 \end{align}   
%Invoking It\^{o}'s formula \cite[Thoerem 7.14]{BMbook10} and \eqref{eq_bound}, yields that almost surely
%\begin{align}\label{eq_ito}
%&u(s+t, b+W_t) - u(s, b) \nonumber\\
%=& \int_0^t \frac{\partial}{\partial b} u(s+r,b+W_{r}) dW_{r}  \nonumber\\
%  &+ \int_0^t \frac{\partial}{\partial s} u(s+r,b+W_{r}) d r \nonumber\\
%  &+ \frac{1}{2}\int_0^t \frac{\partial^2 }{\partial b^2} u(s+r,b+W_{r}) d r. 
% \end{align}
%If $b^2 \geq \beta$, $\frac{\partial}{\partial b} u(s,b) = $; otherwise, if  $b^2 < \beta$,
%$\frac{\partial}{\partial b} u(s,b) = - 2\alpha b$. Hence, for all $t\geq 0$ and $x = (s,b)\in \mathbb{R}^2$, we can obtain
%\begin{align}\label{eq_bound}
%&\mathbb{E}_x\int_0^t \left[\frac{\partial}{\partial b} u(s+r,b+W_{r})\right]^2 dr  <\infty.
%\end{align}
%Invoking It\^{o}'s formula \cite[Thoerem 7.14]{BMbook10} and \eqref{eq_bound}, yields that almost surely
%\begin{align}\label{eq_ito}
%&u(s+t, b+W_t) - u(s, b) \nonumber\\
%=& \int_0^t \frac{\partial}{\partial b} u(s+r,b+W_{r}) dW_{r}  \nonumber\\
%  &+ \int_0^t \frac{\partial}{\partial s} u(s+r,b+W_{r}) d r \nonumber\\
%  &+ \frac{1}{2}\int_0^t \frac{\partial^2 }{\partial b^2} u(s+r,b+W_{r}) d r. 
% \end{align}
%Further, by \eqref{eq_bound} and Thoerem 7.11 in \cite{BMbook10}, the It\^{o}'s integral $\int_0^t \frac{\partial}{\partial b} u(s+r,b+W_{r}) dW_{r}$ is a martingale and 
%\begin{align}\label{eq_ito_integral}
%\mathbb{E}_x\int_0^t \frac{\partial}{\partial b} u(s+r,b+W_{r}) dW_{r} =0,~ \forall~ t\geq0.
% \end{align}  
By combining \eqref{eq_Markov}, \eqref{eq_ito}, and \eqref{eq_ito_integral}, we get
\begin{align} \label{eq_difference}
\mathbb{E}_x \left[u(X_t)\right] \!-\! u(x) = \mathbb{E}_x \left\{\!\int_0^t \!\left[\frac{\partial}{\partial s} u(X_r) \!+\!\frac{1}{2} \frac{\partial^2 }{\partial b^2} u(X_r)\right]\!dr\right\}.
\end{align} 
It is easy to compute that if $b^2> \beta$, 
\begin{align} %\label{eq_gradient}
\frac{\partial}{\partial s} u(s,b) + \frac{1}{2}\frac{\partial^2 }{\partial b^2} u(s,b) =\beta -  3b^2 \leq 0; \nonumber
 \end{align}
and if $b^2< \beta$, 
\begin{align} %\label{eq_gradient_2}
\frac{\partial}{\partial s} u(s,b) + \frac{1}{2}\frac{\partial^2 }{\partial b^2} u(s,b) =\beta -  \beta = 0.\nonumber
 \end{align} 
 Hence, 
 \begin{align} \label{eq_gradient}
 \frac{\partial}{\partial s} u(s,b) + \frac{1}{2}\frac{\partial^2 }{\partial b^2} u(s,b) \leq 0
  \end{align} 
 for all $(s,b)\in \mathbb{R}^2$ except for $b =  \pm\sqrt{\beta}$. Since the Lebesgue measure of those $r$ for which $b+W_{r} = \pm\sqrt{\beta}$ is zero,  we get from \eqref{eq_difference} and \eqref{eq_gradient} that
%By substituting \eqref{eq_gradient} and  \eqref{eq_gradient_2} into \eqref{eq_difference}, we can obtain 
$\mathbb{E}_x \left[u(X_t)\right] \leq u(x)$ for all $x\in \mathbb{R}^2$ and $t\geq 0$. This completes the proof.

\section{Proof of Theorem  \ref{thm_zero_gap}}\label{app_zero_gap}

Theorem  \ref{thm_zero_gap} is proven in three steps:

\emph{Step 1: We will show that the duality gap between \eqref{eq_SD} and \eqref{eq_dual} is zero, i.e., $d(c) = p(c)$.} 

To that end, we needs to find $\pi^{\star} = (Z_0,Z_1,\ldots)$ and $\lambda^{\star}$ that 
%show that \emph{the duality gap between \eqref{eq_SD} and \eqref{eq_dual} is zero, i.e., $d(c) = p(c)$}.
satisfy the following conditions:
\begin{align}
&\pi^{\star}\in\Pi, \lim_{n\rightarrow \infty} \frac{1}{n} \sum_{i=0}^{n-1} \mathbb{E}\left[Y_i+Z_i\right] - \frac{1}{f_{\max}} \geq  0,\label{eq_mix_or_not0}\\
&\lambda^{\star}\geq 0,\label{eq_mix_or_not1}\\
&L(\pi^{\star};\lambda^{\star},c_{\text{opt}}) = \inf_{\pi\in\Pi_1}  L(\pi;\lambda^{\star},c_{\text{opt}}),\label{eq_mix_or_not}\\
&\lambda^\star \left\{\lim_{n\rightarrow \infty} \frac{1}{n} \sum_{i=0}^{n-1} \mathbb{E}\left[Y_i+Z_i\right] - \frac{1}{f_{\max}}\right\} = 0.\label{eq_KKT_last}
\end{align}

%\begin{align}
%&\pi^{\star}\in\Pi, \lim_{n\rightarrow \infty} \frac{1}{n} \sum_{i=0}^{n-1} \mathbb{E}\left[Y_i+Z_i\right] - \frac{1}{f_{\max}} \geq  0,\label{eq_mix_or_not0}\\
%&\lambda^{\star}\geq 0,\label{eq_mix_or_not1}\\
%&L(\pi^{\star};\lambda^{\star},c) = \inf_{\pi\in\Pi_1}  L(\pi;\lambda^{\star},c),\label{eq_mix_or_not}\\
%&\lambda^\star \left\{\lim_{n\rightarrow \infty} \frac{1}{n} \sum_{i=0}^{n-1} \mathbb{E}\left[Y_i+Z_i\right] - \frac{1}{f_{\max}}\right\} = 0.\label{eq_KKT_last}
%\end{align}

According to Theorem \ref{thm_solution_form} and Corollary \ref{coro_stop}, the  solution $\pi^{\star}$ to
\eqref{eq_mix_or_not} is given by \eqref{eq_opt_stop_solution1} where $\beta = 3(c + \lambda^\star - \mathbb{E}\left[Y \right])$. 
In addition, as shown in the proof of Theorem \ref{thm_solution_form}, the $Z_i$'s in policy $\pi^{\star}$ are  \emph{i.i.d.}
From \eqref{eq_mix_or_not0}, \eqref{eq_mix_or_not1}, and \eqref{eq_KKT_last}, $\lambda^\star$ is determined by considering two cases: If $\lambda^\star >0$, then 
\begin{align}\label{eq_KKT_2}
\lim_{n\rightarrow \infty} \frac{1}{n} 
\sum_{i=0}^{n-1} \mathbb{E}\left[Y_i+Z_i\right] = \mathbb{E}\left[Y_i+Z_i\right]= \frac{1}{f_{\max}}.
\end{align}
If $\lambda^\star =0$, then 
\begin{align}\label{eq_KKT_3}
\lim_{n\rightarrow \infty} \frac{1}{n}\sum_{i=0}^{n-1} \mathbb{E}\left[Y_i+Z_i\right] = \mathbb{E}\left[Y_i+Z_i\right] \geq \frac{1}{f_{\max}}.
\end{align}
Hence, such $\pi^{\star}$ and $\lambda^\star$  satisfy \eqref{eq_mix_or_not0}-\eqref{eq_KKT_last}. By \cite[Prop. 6.2.5]{Bertsekas2003}, $\pi^{\star}$ is an optimal solution to the primal problem \eqref{eq_SD} and $\lambda^\star$ is a geometric multiplier \cite{Bertsekas2003} for the primal problem \eqref{eq_SD}. 
In addition, the duality gap between \eqref{eq_SD} and \eqref{eq_dual} is zero, because otherwise there exists no geometric multiplier \cite[Section 6.1-6.2]{Bertsekas2003}. %Hence, the duality gap between \eqref{eq_SD} and \eqref{eq_dual} is zero.

\emph{Step 2: We will show that a common optimal solution to \eqref{eq_DPExpected},  \eqref{eq_Simple}, and \eqref{eq_SD} with \emph{$c = c_{\text{opt}}$}, is given by \eqref{eq_opt_solution} where 
$\beta  \geq0$ is determined by solving
\begin{align}\label{eq_equation}
\mathbb{E}[Y_i+Z_i] \!=\! \max\left(\frac{1}{f_{\max}}, \frac{\mathbb{E}[(W_{S_{i}+Y_i+Z_i}-W_{S_i})^4]}{2\beta}\right),\!
\end{align}}
where the last term is determined by L'H\^{o}pital's rule if $\beta\rightarrow0$.

%that if \emph{$c =  c_{\text{opt}}$}, $\pi^{\star}$ is given by \eqref{eq_opt_solution} and \emph{$\beta=3(c + \lambda^\star - \mathbb{E}\left[Y \right])$} is determined by solving \eqref{eq_thm1}.} 

We  consider the case that $c =  c_{\text{opt}}$. In \emph{Step 1}, we have shown that policy $\pi^{\star}$ in \eqref{eq_opt_stop_solution1} with $\beta = 3(c_{\text{opt}} + \lambda^\star - \mathbb{E}\left[Y \right])$ is an optimal solution to  \eqref{eq_SD}. According to the definition of  {$c_{\text{opt}}$} in \eqref{eq_c}, $p(c_{\text{opt}})=0$. By Lemma \ref{lem_ratio_to_minus}(b), this policy $\pi^{\star}$ is also an optimal solution to \eqref{eq_Simple}. In addition, $p(c_{\text{opt}})=0$ and Lemma \ref{lem_ratio_to_minus}(a) imply ${\mathsf{mmse}}_{\text{opt}} = c_{\text{opt}} $.  
Substituting policy $\pi^{\star}$ and \eqref{eq_integral} into \eqref{eq_Simple}, yields
\begin{align}\label{eq_KKT_1}
c_{\text{opt}} &=\lim_{n\rightarrow \infty}\!\!\frac{\sum_{i=0}^{n-1}\mathbb{E}\left[(W_{S_{i}+Y_i+Z_i}\!-\!W_{S_i})^4 \!+\! (Y_i+Z_i) \mathbb{E}[Y] \right]}{ 6\sum_{i=0}^{n-1} \mathbb{E}\left[Y_i\!+\!Z_i\right]} \nonumber\\
&=\frac{\mathbb{E}\left[(W_{S_{i}+Y_i+Z_i}\!-\!W_{S_i})^4\right]}{ 6\mathbb{E}\left[Y_i\!+\!Z_i\right]} + \mathbb{E}[Y],
\end{align}
where in the last equation we have used that the $Z_i$'s are \emph{i.i.d.} and the $(W_{S_{i}+Y_i+Z_i}-W_{S_i})$'s are \emph{i.i.d.}, which were shown  in the proof of Theorem \ref{thm_solution_form}. According to \eqref{eq_KKT_1}, $c_{\text{opt}}\geq \mathbb{E}[Y]$. Hence,  $\beta = 3(c_{\text{opt}} + \lambda^\star - \mathbb{E}\left[Y \right])\geq 0$, in which case policy $\pi^{\star}$ in \eqref{eq_opt_stop_solution1} is exactly \eqref{eq_opt_solution}. 

The value of $\beta$ can be obtained by considering the following two cases:
% (for notational simplicity, we remove the superscript ``$\star$'' in the sequel):

\emph{Case 1}: If $\lambda >0$, then \eqref{eq_KKT_1} and \eqref{eq_KKT_2} imply that
\begin{align}\label{eq_KKT_4}
&\mathbb{E}\left[Y_i+Z_i\right]= \frac{1}{f_{\max}}, \\
&\beta > 3(c_{\text{opt}} -\mathbb{E}[Y])  = \frac{\mathbb{E}\left[(W_{S_{i}+Y_i+Z_i}-W_{S_i})^4\right]}{ 2\mathbb{E}\left[Y_i\!+\!Z_i\right]}.
\end{align}

\emph{Case 2}: If $\lambda =0$, then \eqref{eq_KKT_1} and \eqref{eq_KKT_3} imply that
\begin{align}
&\mathbb{E}\left[Y_i+Z_i\right]\geq \frac{1}{f_{\max}}, \\
&\beta = 3(c_{\text{opt}} -\mathbb{E}[Y]) = \frac{\mathbb{E}\left[(W_{S_{i}+Y_i+Z_i}-W_{S_i})^4\right]}{ 2\mathbb{E}\left[Y_i\!+\!Z_i\right]}.\label{eq_KKT_5}
\end{align}
Combining \eqref{eq_KKT_4}-\eqref{eq_KKT_5}, \eqref{eq_equation} follows. By \eqref{eq_KKT_1}, the optimal value of \eqref{eq_Simple} is  given by 
\begin{align}\label{thm_3_obj}
\mathsf{mmse}_{\text{opt}}=\frac{\mathbb{E}[(W_{S_{i}+Y_i+Z_i}-W_{S_i})^4]}{6\mathbb{E}[Y_i+Z_i] } + \mathbb{E}[Y].
\end{align}

\emph{Step 3: We will show that the expectations in \eqref{eq_equation} and \eqref{thm_3_obj} 
are given by
\begin{align}
&\mathbb{E}[Y_i+Z_i] = \mathbb{E}[\max(\beta,W_Y^2)],\label{eq_expression0}\\
&\mathbb{E}[(W_{S_{i}+Y_i+Z_i}-W_{S_i})^4] = \mathbb{E}[\max(\beta^2,W_Y^4)].\label{eq_expression}
\end{align}}
According to \eqref{eq_opt_stop_solution1} with $\beta\geq 0$, we have
\begin{align}
&W_{S_{i}+Y_i+Z_i}-W_{S_{i}} \nonumber\\
=&\left\{\begin{array}{l l}
W_{S_{i}+Y_i}-W_{S_{i}},&\text{if}~ |W_{S_{i}+Y_i}-W_{S_{i}}| \geq \sqrt{\beta};\\
\sqrt{\beta},&\text{if}~ |W_{S_{i}+Y_i}-W_{S_{i}}| < \sqrt{\beta}.
\end{array}\right. \nonumber
\end{align}
Hence, 
\begin{align}\label{eq_max_1}
\mathbb{E}[(W_{S_{i}+Y_i+Z_i}-W_{S_{i}})^4] = \mathbb{E}[\max(\beta^2,(W_{S_{i}+Y_i}-W_{S_{i}})^4)].
\end{align}
%which implies \eqref{eq_expression}.
In addition, from \eqref{eq_zero} and \eqref{eq_expectation} we know that if  $|W_{S_{i}+Y_i}-W_{S_{i}}| \geq \sqrt{\beta}$
\begin{align}
\mathbb{E}[Z_i| Y_i] = 0;\nonumber
\end{align}
otherwise,
\begin{align}
\mathbb{E}[Z_i| Y_i] = \beta - (W_{S_{i}+Y_i}-W_{S_{i}})^2.\nonumber
\end{align}
Hence,
\begin{align}
\mathbb{E}[Z_i| Y_i] = \max[\beta - (W_{S_{i}+Y_i}-W_{S_{i}})^2,0].\nonumber
\end{align}
Using the law of iterated expectations, the strong Markov property of the Wiener process, and  Wald's identity $\mathbb{E}[(W_{S_{i}+Y_i}-W_{S_{i}})^2]=\mathbb{E}[Y_i]$, yields
\begin{align}\label{eq_max_2}
&\mathbb{E}[Z_i+ Y_i] \nonumber\\
=&\mathbb{E}[ \mathbb{E}[Z_i| Y_i] +Y_i]\nonumber\\
=&\mathbb{E}[ \max(\beta - (W_{S_{i}+Y_i}-W_{S_{i}})^2,0) +Y_i] \nonumber\\
=&\mathbb{E}[ \max(\beta - (W_{S_{i}+Y_i}-W_{S_{i}})^2,0) +(W_{S_{i}+Y_i}-W_{S_{i}})^2] \nonumber\\
=&\mathbb{E}[ \max(\beta,(W_{S_{i}+Y_i}-W_{S_{i}})^2)]. 
\end{align}
Finally, because $W_t$ and $W_{S_{i}+t}-W_{S_{i}}$ are of the same distribution,  \eqref{eq_expression0} and \eqref{eq_expression} follow from \eqref{eq_max_2} and \eqref{eq_max_1}, respectively. 
This completes the proof.

\end{document}